\documentclass[journal,onecolumn, 12pt]{IEEEtran}

\usepackage{amssymb}
\usepackage{amsmath}
\usepackage{graphicx}
\usepackage{placeins}
\usepackage{mathrsfs}
\usepackage{bbm}
\usepackage{marvosym}

\usepackage[permil]{overpic}

\usepackage{todonotes}
\usepackage[normalem]{ulem}
\usepackage{dcolumn}
\usepackage{hyperref}
\usepackage{rotating}
\usepackage{mathtools}
\usepackage[T1]{fontenc}
\usepackage{caption}
\usepackage{accents}
\newcommand\thickbar[1]{\accentset{\rule{.65em}{.8pt}}{#1}}

\usetikzlibrary{decorations.pathreplacing}
\usetikzlibrary{arrows}
\usetikzlibrary{shapes,snakes}
\usetikzlibrary{calc}
\usetikzlibrary{decorations}
\usetikzlibrary{fadings}
\tikzfading[name=fade out, 
    inner color=transparent!0,
    outer color=transparent!100]
\pgfdeclaredecoration{lightning bolt}{draw}{
\state{draw}[width=\pgfdecoratedpathlength]{
  \pgfpathmoveto{\pgfpointorigin}%
  \pgfpathlineto{\pgfpoint{\pgfdecoratedpathlength*0.6}%
    {-\pgfdecoratedpathlength*.1}}%
  \pgfpathlineto{\pgfpoint{\pgfdecoratedpathlength*0.55}{0pt}}%
  \pgfpathlineto{\pgfpoint{\pgfdecoratedpathlength}{0pt}}%
  \pgfpathlineto{\pgfpoint{\pgfdecoratedpathlength*0.4}%
    {\pgfdecoratedpathlength*.1}}%
  \pgfpathlineto{\pgfpoint{\pgfdecoratedpathlength*0.45}{0pt}}%
  \pgfpathclose%
}%
}

\makeatletter
\newcommand*{\centerfloat}{%
  \parindent \z@
  \leftskip \z@ \@plus 1fil \@minus \textwidth
  \rightskip\leftskip
  \parfillskip \z@skip}
\makeatother

\usepackage[braket]{qcircuit}

\usepackage{amsthm} 
    \newtheorem{theorem}{Theorem}[section]
    \newtheorem{lemma}[theorem]{Lemma}
    \newtheorem{corollary}[theorem]{Corollary}
    \newtheorem{proposition}[theorem]{Proposition}
    \theoremstyle{definition}
    \newtheorem{definition}[theorem]{Definition}

\usepackage{physics}
\usepackage{qcircuit}
\usepackage{subcaption}
\usepackage{enumitem}

\DeclareMathOperator{\Sp}{Sp}

\newcommand{\F}{\mathbb{F}}

\newcommand{\transp}{\mathsf{T}}

\newcommand{\CNOT}{\mathrm{CNOT}}
\newcommand{\CZ}{\mathrm{CZ}}
\newcommand{\SWAP}{\mathrm{SWAP}}
\newcommand{\Hgate}{H}
\newcommand{\Sgate}{S}
\newcommand{\ssep}{\mid}

\begin{document}

\title{
Near-term $n$ to $k$ distillation protocols using graph codes
}


\author{
  \IEEEauthorblockN{Kenneth Goodenough\IEEEauthorrefmark{1}\IEEEauthorrefmark{2}, S\'ebastian de Bone\IEEEauthorrefmark{2}, Vaishnavi Addala\IEEEauthorrefmark{4}, Stefan Krastanov\IEEEauthorrefmark{2} \IEEEauthorrefmark{4}, Sarah Jansen\IEEEauthorrefmark{1}\IEEEauthorrefmark{6}, Dion Gijswijt\IEEEauthorrefmark{5}, David Elkouss\IEEEauthorrefmark{1} \IEEEauthorrefmark{7}}\\
  \vspace{2mm}
    \IEEEauthorblockA{\IEEEauthorrefmark{1}QuTech, Delft University of Technology}
    
    \IEEEauthorblockA{\IEEEauthorrefmark{2}College of Information and Computer Science, University of Massachusetts Amherst}

        \IEEEauthorblockA{\IEEEauthorrefmark{3}QuSoft, CWI}

        \IEEEauthorblockA{\IEEEauthorrefmark{4}Department of Electrical Engineering and Computer Science, Massachusetts Institute of Technology}

        \IEEEauthorblockA{\IEEEauthorrefmark{5}Delft Institute of Applied Mathematics, Delft University of Technology}

                \IEEEauthorblockA{\IEEEauthorrefmark{6}Korteweg-de Vries Institute for Mathematics, University of Amsterdam}

                 \IEEEauthorblockA{\IEEEauthorrefmark{7}Networked Quantum Devices Unit, Okinawa Institute of Science and Technology Graduate University}
    }

\maketitle

\begin{abstract}
Noisy hardware forms one of the main hurdles to the realization of a near-term quantum internet. Distillation protocols allows one to overcome this noise at the cost of an increased overhead. We consider here an experimentally relevant class of distillation protocols, which distill $n$ to $k$ end-to-end entangled pairs using bilocal Clifford operations, a single round of communication and a possible final local operation depending on the observed measurement outcomes. In the case of permutationally invariant depolarizing noise on the input states, we find a correspondence between these distillation protocols and graph codes. We leverage this correspondence to find provably optimal distillation protocols in this class for several tasks important for the quantum internet. This correspondence allows us to investigate use cases for so-called non-trivial measurement syndromes. Furthermore, we detail a recipe to construct the circuit used for the distillation protocol given a graph code. We use this to find circuits of short depth and small number of two-qubit gates.
Additionally, we develop a black-box circuit optimization algorithm, and find that both approaches yield comparable circuits.
Finally, we investigate the teleportation of encoded states and find protocols which jointly improve the rate and fidelities with respect to prior art.
\end{abstract}

\section{Introduction}
Entanglement is a key feature of quantum mechanics, and is the fundamental resource to be distributed in the quantum internet. Unfortunately, experimental setups are imperfect, leaving entanglement noisy in practice. Entanglement distillation is any procedure using local operations and classical communication that (usually probabilistically) converts $n$ input states to (usually) a smaller number of states $k$ with increased fidelity~\cite{bennett1996purification,Bennett1996,Deutsch1996, dur2007entanglement}. Distillation thus allows for overcoming the effects of inherent noise in any physical implementation of a quantum network.

Finding good distillation protocols that are also feasible experimentally is thus important for the workings of future quantum networks~\cite{krastanov2019optimized, Rozpdek2018}. This motivates us to study distillation protocols that 1) distill from $n$ to $k$ pairs for $n$ relatively small, i.e.~$n\lesssim 10$, 2) require only a single round of communication, and 3) use only operations that are relatively simple to implement. For the latter, we allow both parties to apply operations of the form $C^\transp\otimes C^\dagger$, where $C$ is a Clifford circuit, i.e.~constructed from $H$, $S$ and $\CNOT$ gates. Such Clifford circuits are relevant since they form a key component for quantum applications and can be efficiently implemented~\cite{Bravyi2020}. Furthermore, all but the first $k$ pairs are measured in the computational basis, after which a final operation conditioned on the measurement outcomes is allowed. Specific instances of such \emph{bilocal Clifford protocols} have been considered in the literature~\cite{bennett1996purification,Deutsch1996,fujii2009entanglement,briegel1998quantum,dur2003entanglement,dur1999quantum,ruan2018adaptive,vollbrecht2005interpolation,krastanov2019optimized, jansen2020enum}.


\begin{figure*}
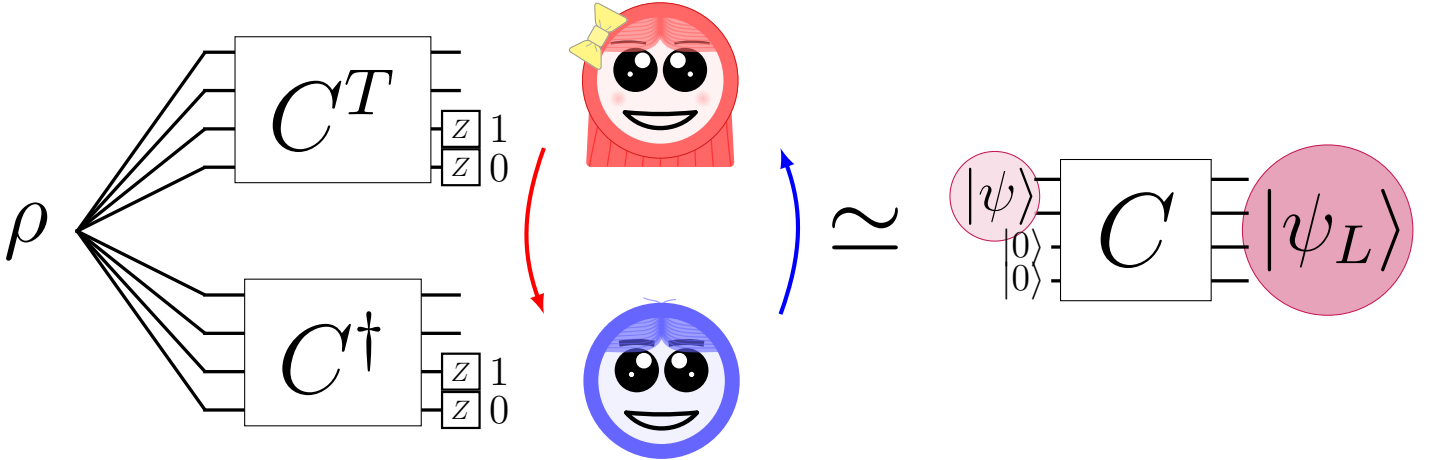

\centerfloat
\hspace{-10mm}
    \hspace{-22mm}
    \begin{subfigure}[htp]{0.4\textwidth}
\begin{tikzpicture}[scale=1.7]
\node[scale=2.8] at (-0.4,0){$\rho$};
\node[scale=4.7] at (4.5,1.15){\input{alice.tex}};
\node[scale=4.7] at (4.57,-1.15){\input{bob.tex}};

\filldraw[color=black!100, fill=black!100, very thick](0,0) circle (0.005);

\draw[line width=1.25] (0,0) -- (1,0.5);
\draw[line width=1.25] (0,0) -- (1,0.8);
\draw[line width=1.25] (0,0) -- (1,1.1);
\draw[line width=1.25] (0,0) -- (1,1.4);

\draw[line width=1.25] (3,1/2) -- (0.985,1/2);
\draw[line width=1.25] (3,0.8) -- (0.985,0.8);
\draw[line width=1.25] (3,1.1) -- (0.985,1.1);
\draw[line width=1.25] (3,1.4) -- (0.985,1.4);

\node[scale=2.10] at (3,1/2){\input{meas.tex}};
\node[scale=1.3] at (3.3,1/2){$0$};
\node[scale=2.10] at (3,0.8){\input{meas.tex}};
\node[scale=1.3] at (3.3,0.8){$1$};

\draw[line width=1.25] (0,0) -- (1,-0.5);
\draw[line width=1.25] (0,0) -- (1,-0.8);
\draw[line width=1.25] (0,0) -- (1,-1.1);
\draw[line width=1.25] (0,0) -- (1,-1.4);

\draw[line width=1.25] (3,-1/2) -- (0.985,-1/2);
\draw[line width=1.25] (3,-0.8) -- (0.985,-0.8);
\draw[line width=1.25] (3,-1.1) -- (0.985,-1.1);
\draw[line width=1.25] (3,-1.4) -- (0.985,-1.4);

\node[scale=2.10] at (3,-1.1){\input{meas.tex}};
\node[scale=1.3] at (3.3,-1.1){$1$};
\node[scale=2.10] at (3,-1.4){\input{meas.tex}};
\node[scale=1.3] at (3.3,-1.4){$0$};

\node[draw, scale=3.1, fill=white] at (2.0,0.95){$C^{T}$};
\node[draw, scale=3.1, fill=white] at (2.0,-0.95){$C^{\dagger}$};
\node[anchor=east] at (5.5, 0) (node1){};
\node[anchor=east] at (5.5, -1.25) (node2){};

\draw [-latex,red, line width=0.55mm] (3.65,0.65) to [out=270-20,in=180-70] (3.65,-0.65);
\draw [-latex,blue, line width=0.55mm] (5.5,-0.65) to [in=290,out=70] (5.5,0.65);

\end{tikzpicture}
\end{subfigure}\hspace{35mm}
\begin{subfigure}[htp]{0.3\textwidth}
\begin{tikzpicture}[scale=1.5]

\draw[line width=1.25] (3,1/2) -- (1.25,1/2);
 \node[ellipse, draw, scale=5.7, draw=purple!90, fill=purple!45, fill opacity=0.8] at (3.7,0.95){};
\node[scale=2.6] at (3.75,0.95){$\ket{\psi_L}$};
\draw[line width=1.25] (3,0.8) -- (1.25,0.8);
\draw[line width=1.25] (3,1.1) -- (1.1,1.1);
 \node[ellipse, draw, scale=3.0, draw=purple!90, fill=purple!15, fill opacity=0.8] at (0.75,1.25){};
\node[scale=1.8] at (0.8,1.25){$\ket{\psi}$};
\draw[line width=1.25] (3,1.4) -- (1.1,1.4);
\node[scale=1.2] at (1,1/2){$\ket{0}$};
\node[scale=1.2] at (1,0.8){$\ket{0}$};

\node[draw=none, scale=3.1] at (-0.4,0.95){$\simeq$};

\node[draw, scale=3.3, fill=white] at (2.0,0.95){$C$};
\end{tikzpicture}
\end{subfigure}
\caption{Correspondence between bilocal Clifford distillation protocols and stabilizer codes. On the left we show the general form bilocal Clifford distillation protocols can take. That is, Alice and Bob apply $C^\transp$ and $C^\dagger$ for some Clifford circuit $C$, and then measure out the last $n-k$ pairs. They then use classical communication to send the measurement outcomes to one another, and use those to decide on whether to keep the states and/or apply a final correction. On the right we show a stabilizer code, which takes in a state $\ket{\psi}$, and transforms it to a logical state $\ket{\psi_L}$ by applying a Clifford circuit $C$ to $\ket{\psi}$ and $n-k$ auxiliary qubits. There is a one-to-one correspondence between stabilizer codes bilocal Clifford distillation protocols and stabilizer codes, given by using a fixed Clifford circuit $C$ in both cases.}
\label{fig:intro_fig1}
\end{figure*}

Our goal is to find good near-term bilocal Clifford distillation protocols. To this end, we use two methods. Firstly, an approach based on graph theory to find provably optimal (with respect to any measure) bilocal Clifford protocols in the case of uniformly depolarized states and no noisy operations. Secondly, an approach based on black-box optimization with genetic algorithms~\cite{krastanov2019optimized}. This framework is flexible, allowing for a heuristic optimization even when considering arbitrary Pauli noise, noisy circuits and limitations on the number of qubits that can be simultaneously processed.

The graph-theoretical framework reduces the optimization over bilocal Clifford protocols to a smaller set of certain equivalence classes on graphs of $n+k$ vertices. The number of equivalence classes is significantly smaller than the number of possible Clifford circuits, allowing us to optimize by performing a full enumeration.

We compare circuits found using the graph-theoretical approach and with the black-box algorithm. We find that both approaches yield similar results, where each approach works best in different parameter regimes.
Finally, we consider the procedure of teleporting and correcting encoded states. This requires two parties to share a bipartite state of local dimension~$2^k$. These states can be generated in multiple ways. Here, we consider creating $k$ bipartite states, creating $k$ distilled bipartite states out of $2k$ states through use of the DEJMPS protocol~\cite{Deutsch1996}, or by distilling once $n$ pairs to $k$ pairs. We find that the latter option can provide higher fidelities and success probabilities, while also using fewer resources than distilling $k$ pairs independently. 

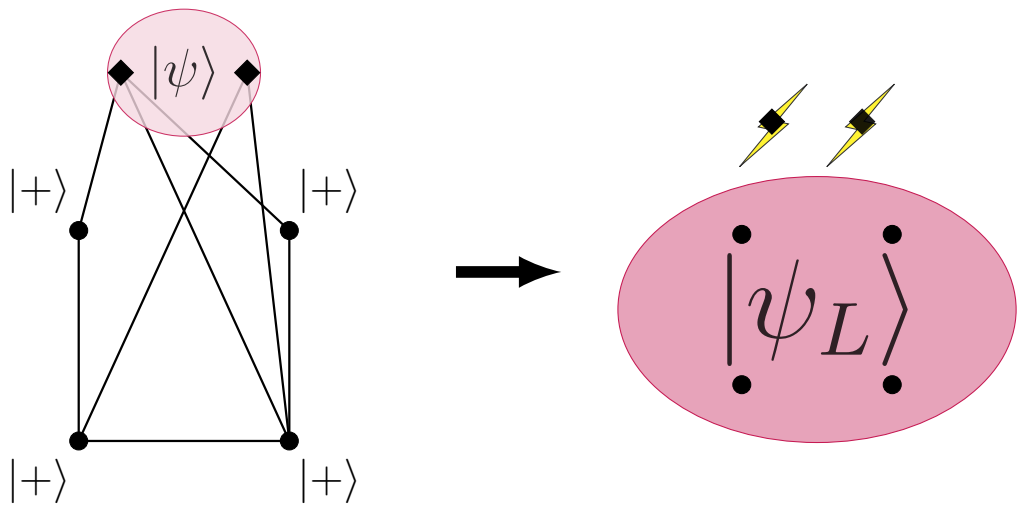
\begin{figure*}
\centerfloat
\hspace{-20mm}
\begin{subfigure}[htp]{0.3\textwidth}
\begin{tikzpicture}[scale=1.4]

\node[circle, fill=black, draw, scale=0.6] (a) at (-1,+1){};
\node[circle, fill=black, draw, scale=0.6] (b) at (+1, +1){};
\node[circle, fill=black, draw, scale=0.6] (c) at (+1, -1){};
\node[circle, fill=black, draw, scale=0.6] (d) at (-1, -1){};

\draw[line width = 0.3mm] (1,1) -- (1,-1) -- (-1,-1) -- (-1,+1);

\draw[line width = 0.3mm] (-0.6, 2.5) -- (-1,1);
\draw[line width = 0.3mm] (-0.6, 2.5) -- (1,-1);
\draw[line width = 0.3mm] (-0.6, 2.5) -- (1,1);

\draw[line width = 0.3mm] (0.6, 2.5) -- (-1,-1);
\draw[line width = 0.3mm] (0.6, 2.5) -- (1,-1);

  \node[ellipse, draw, scale=1.7, draw=purple!80, fill=purple!15, fill opacity=0.8] (e) at (0,2.5) {$\ket{\psi}$};

\node[diamond, fill=black, draw, scale=0.6] (d1) at (0.6, 2.5){};
\node[diamond, fill=black, draw, scale=0.6] (d2) at (-0.6, 2.5){};

\node[scale=1.4] (d) at (-1-0.38, +1+0.38){$\ket{+}$};
\node[scale=1.4] (d) at (+1+0.38, +1+0.38){$\ket{+}$};
\node[scale=1.4] (d) at (-1-0.38, -1-0.38){$\ket{+}$};
\node[scale=1.4] (d) at (+1+0.38, -1-0.38){$\ket{+}$};





\end{tikzpicture}
\end{subfigure}\hspace{5mm}
\begin{subfigure}[htp]{0.3\textwidth}
\begin{tikzpicture}

\draw [-latex,black, line width=1.55mm] (-4.8,0.5) to (-3.4,0.5);




  \node[ellipse, draw, scale=3.56, draw=purple!90, fill=purple!45, fill opacity=0.8] (e) at (0,0) {$\vspace*{5mm}\ket{\psi_L}\vspace*{5mm}$};

\fill [yellow, opacity=0.8, draw=black, decoration=lightning bolt, decorate] 
  (0.13, 1.90) --  (1.03, 3.00);
  \fill [yellow, opacity=0.8, draw=black, decoration=lightning bolt, decorate] 
  (-1.03, 1.90) --  (-0.13, 3.00);
\node[diamond, fill=black, opacity=0.9, draw, scale=0.6] (d1) at (0.6, 2.5){};
\node[diamond, fill=black, draw, scale=0.6] (d2) at (-0.6, 2.5){};

\node[circle, fill=black, draw, scale=0.6] (a) at (-1,+1){};
\node[circle, fill=black, draw, scale=0.6] (b) at (+1, +1){};
\node[circle, fill=black, draw, scale=0.6] (c) at (+1, -1){};
\node[circle, fill=black, draw, scale=0.6] (d) at (-1, -1){};






\end{tikzpicture}
\end{subfigure}
\caption{Here we show an alternative approach to how one could implement a subset of the stabilizer encodings. That is, first prepare the $k$ qubit input state (the corresponding qubits are called input vertices). Then prepare $n$ output qubits in the $\ket{+}$ state. Then, $\CZ$ gates are applied according to some simple graph on $n+k$ vertices, where we distinguish between the in- and output vertices. Such objects we call $\left(n, k\right)$-graphs.
Then, on the right the input qubits are measured in the $X$-basis, initializing the remaining $n$ output qubits in some logical state $\ket{\psi_L}$. When correcting against depolarizing noise, it suffices to consider encodings performed in this way~\cite{schlingemann2001stabilizer}. This thus reduces the optimization to one over $\left(n, k\right)$-graphs. Finally, we reduce the search space even further by showing that $\left(n, k\right)$-graphs that are equivalent under so-called \emph{local complementations}, \emph{edge flips} and (in the case of permutationally invariant depolarizing noise) permutations of the input vertices and permutations of the output vertices yield equivalent distillation protocols. We note that the $\left(n,k\right)$-graph formalism can also be used to construct circuits that implement the corresponding distillation protocols/stabilizer codes (not shown in this figure).}
\label{fig:intro_fig2}
\end{figure*}

The rest of this work is structured as follows. We start by laying down the preliminaries and the used notation in Section~\ref{sec:prelim}. In Section~\ref{sec:correspondence} we detail explicitly the correspondence between stabilizer codes and bilocal Clifford protocols. We specialize this correspondence to the case of distilling an $n$-fold tensor power of a Werner state in Section~\ref{sec:graphreduc}. This allows us to study bilocal Clifford distillation protocols through the study of graph codes. In particular, we show it is possible to find all bilocal Clifford distillation protocols on an $n$-fold tensor power of a Werner state for several values of $n$ and $k$ by searching over all graph codes.
In Section~\ref{sec:circuits}, we detail a way to convert a bilocal Clifford distillation protocol via a corresponding graph code into a circuit. We then discuss certain heuristics that can be used to improve circuits (such as reducing the depth) given a graph code.
Given a circuit of a distillation protocol, we discuss briefly how to calculate the quantities of interest in Section~\ref{sec:simpcalc}. These quantities are the probability and the coefficients of the output state as a function of the observed measurements.
Using the above tools, we analyse the performance of our found protocols for several communication tasks/metrics in Section~\ref{sec:results}. We end with concluding remarks and potential avenues for further research in Section~\ref{sec:conclusions}.

\section{Preliminaries}\label{sec:prelim}
Here we set our used notation and definitions, most of which is similar to the notation in~\cite{jansen2020enum}. We denote by $\F_2$ the field with two elements. Relevant single-qubit operations are given by the Pauli operators $I, X, Y, Z$, Hadamard gate $\Hgate$ and phase gate $\Sgate$. A subscript indicates a specific qubit, e.g.~$H_2$ denotes a Hadamard gate acting on the second qubit and the identity $I$ acting on the remaining qubits, where we assume there is an ordering given on the qubits. We use the term \emph{single-qubit Clifford operations} to refer to the elements in the group generated by Hadamard and phase gates on each qubit.

The relevant two-qubit operations are given by the controlled-not operation $\CNOT_{ij}$, controlled-$Z$ operation $\CZ_{ij}$ and swap operation $\SWAP_{ij}$. For the $\CNOT_{ij}$ operation, the subscripts $i$ and $j$ indicate the control and target, respectively.

The Pauli operators expanded in the computational basis are given by

\begin{equation}
    \begin{split}
    I &= \begin{bmatrix} 1 & 0 \\ 0 & 1 \end{bmatrix}, \\
    Y &= \begin{bmatrix} 0 & -i \\ i & 0 \end{bmatrix},
  \end{split}
\qquad
    \begin{split}
    X &= \begin{bmatrix} 0 & 1 \\ 1 & 0 \end{bmatrix},\\
    Z &= \begin{bmatrix} 1 & 0 \\ 0 & -1 \end{bmatrix}.
    \end{split}
    \label{eq:paulimatrices}
\end{equation}

These single-qubit Pauli operators can be extended to $n$ qubits, yielding the Pauli group $\mathcal{\thickbar P}_n$. The group $\mathcal{ P}_n$ consists of all matrices that are tensor products of Pauli operators, up to phases from $\lbrace{\pm 1, \pm i\rbrace}$. That is, $\mathcal{ P}_n \cong \mathcal{\thickbar P}_n/\langle i I^{\otimes n}\rangle $. With abuse of terminology we will say that two elements of $\mathcal{ P}_n$ (anti-)commute if arbitrary elements in their pre-images (anti-)commute. Note that this is well-defined, since it does not depend on the choice of elements in the preimage.

The \emph{weight} $\textrm{wt}$ of an element of $\mathcal{ P}_n $ is the number of non-identity Pauli elements in the string. For a subset $S$ of $\mathcal{P}_n$, let $\mathcal{E}_w\hspace{-0.5mm}\left(S\right)$ be the number of elements in $S$ with weight $w$. We will refer to the collection of $\mathcal{E}_w\hspace{-0.5mm}\left(S\right)$ as \emph{the weight enumerator of $S$}. Furthermore, define the \emph{weight enumerator polynomial of $S$} as $\mathcal{E}\hspace{-0.5mm}\left(S, x, y\right) = \sum_{w=0}^n \mathcal{E}_w\hspace{-0.5mm}\left(S\right)x^{n-w}y^w$. These objects are related to the weight enumerators used in (quantum) error correction~\cite{gottesman1997stabilizer}, and will turn out to be useful to express the output states of distillation protocols with.

The Clifford group $\mathcal{C}_n$ on $n$ qubits is the group generated by $\Hgate$, $\Sgate$ operations on any qubit, and $\CNOT_{ij}$ between any two qubits $i$ and $j$. The Clifford group acts on $\mathcal{P}_n$ by conjugation, and in fact each automorphism of~$\mathcal{P}_n$ that preserves the commutation relations arises as the conjugation by some $C\in \mathcal{C}_n$.

\subsection{Symplectic representation}
There is a convenient representation of Pauli operators (without phase) and the action of the Clifford group on the Pauli operators in terms of linear algebra over $\F_2$.

Elements of $\mathcal{P}_n$ are represented by elements of $\F_2^{2n}$. In particular, $X_i$ and $Z_i$ are represented by the standard basis vectors $e_{i}$ and $e_{i+n}$, respectively. The representation can then be linearly extended to arbitrary Pauli strings. It can be checked that multiplication in $\mathcal{P}_n$ corresponds to vector addition in $\F_2^{2n}$.

Let $\Omega = \begin{bmatrix} 0 & I_n \\ -I_n & 0 \end{bmatrix}$ and $\omega \colon \F_2^{2n} \times \F_2^{2n} \to \F_2$ be the standard symplectic bilinear form given by
\begin{gather}
    \omega\left(v, w\right) = v^\transp\Omega w \ .
\end{gather}
Two Pauli strings commute iff $\omega$ evaluated on the two corresponding binary vectors $v, w$ equals zero.

Furthermore, conjugation by a Clifford corresponds to a symplectic linear transformation, i.e.~there is a surjective group homomorphism from the Clifford group to the symplectic group of order $n$ over $\F_2$,
\begin{align}
    \Sp(2n,\mathbb{F}_2) = \lbrace M \in \textrm{Mat}_{2n}\left(\mathbb{F}_2)\right) \mid M^\transp\Omega M = \Omega \rbrace \ .
\end{align}
Thus $\Sp(2n,\mathbb{F}_2)$ consists of those matrices $M$ such that $\omega(Mv, Mw) = \omega(v, w)$, $\forall v, w \in \F_2^{2n}$.

\subsection{Graph theory}
We consider here only simple undirected graphs --- that is, graphs with no loops and at most one edge between any two vertices. A graph $G = (V, E)$ has a vertex set~$V$ and edge set~$E$, the latter of which has as elements unordered pairs of vertices. The neighborhood~$N_v$ of a vertex~$v$ is the set of all adjacent vertices of $v$, i.e.~$N_v =  \{\, w\in V \ssep \lbrace{ v, w\rbrace} \in E \}$. Given a subset $S \subseteq V$ of a graph~$G$, the induced subgraph~$G \left[ S \right]$ is defined as the graph with vertex set~$S$ and an edge set containing all edges that are incident with vertices in~$S$ only. Furthermore, $G - S$ is defined as $G\left[ V\setminus S \right]$.

A local complementation $\tau_v$ is an operation on a graph~$G$ that for a vertex~$v$ takes the \emph{graph complement} on the induced subgraph $G\left[N_v\right]$, while leaving the rest of the edges invariant~\cite{bouchet1988graphic}. That is, for each pair of vertices in the neighborhood of~$v$, an edge is added if it was not present, and removed if it was present. We show an example of a local complementation in Fig.~\ref{fig:lcexample}. Two graphs that are related by a sequence of local complementations are \emph{LC equivalent}. These operations will be important to describe operations on representations of distillation protocols.

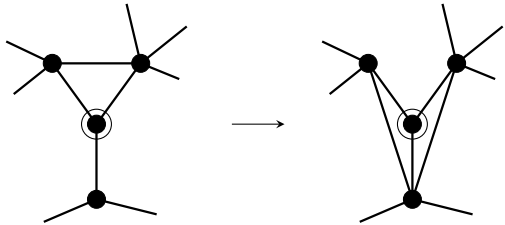
\begin{figure}[h!]
\centerfloat

\begin{tikzpicture}

\node[circle, fill=black, draw, scale=0.6] (1) at ({sin(0*360/5)}, {-cos(0*360/5)}){};
\node[circle, fill=black, draw, scale=0.6] (3) at ({sin(2*360/5)}, {-cos(2*360/5)}){};
\node[circle, fill=black, draw, scale=0.6] (4) at ({sin(3*360/5)}, {-cos(3*360/5)}){};

\node[circle, fill=white, draw, scale=1] (0) at (0.0, -0){};
\node[circle, fill=black, draw, scale=0.6] (0) at (0.0, -0){};




\draw[line width = 0.3mm] ({sin(2*360/5)}, {-cos(2*360/5)}) -- ({sin(3*360/5)}, {-cos(3*360/5)});

\draw[line width = 0.3mm] (0,-1) -- (0.8, -1.2);
\draw[line width = 0.3mm] (0,-1) -- (-0.7, -1.3);

\draw[line width = 0.3mm] ({sin(2*360/5)}, {-cos(2*360/5)}) -- (1.1, 0.6);
\draw[line width = 0.3mm] ({sin(2*360/5)}, {-cos(2*360/5)}) -- (1.2, 1.3);
\draw[line width = 0.3mm] ({sin(2*360/5)}, {-cos(2*360/5)}) -- (0.4, 1.6);
\draw[line width = 0.3mm] ({sin(3*360/5)}, {-cos(3*360/5)}) -- (-1.1, 0.4);
\draw[line width = 0.3mm] ({sin(3*360/5)}, {-cos(3*360/5)}) -- (-1.2, 1.1);

\draw[line width = 0.3mm] (0,-1) -- (0, -0);
\draw[line width = 0.3mm] ({sin(2*360/5)}, {-cos(2*360/5)}) -- (0, -0);
\draw[line width = 0.3mm] ({sin(3*360/5)}, {-cos(3*360/5)}) -- (0, -0);





\draw [-stealth](1.8, 0) -- (2.5,0);



\def\ra{4.2}

\node[circle, fill=black, draw, scale=0.6] (1) at ({sin(0*360/5)+\ra}, {-cos(0*360/5)}){};
\node[circle, fill=black, draw, scale=0.6] (3) at ({sin(2*360/5)+\ra}, {-cos(2*360/5)}){};
\node[circle, fill=black, draw, scale=0.6] (4) at ({sin(3*360/5)+\ra}, {-cos(3*360/5)}){};

\node[circle, fill=white, draw, scale=1] (0) at (0.0+\ra, -0){};
\node[circle, fill=black, draw, scale=0.6] (0) at (0.0+\ra, -0){};






\draw[line width = 0.3mm] ({sin(2*360/5)+\ra}, {-cos(2*360/5)}) -- ({0+\ra}, {-1});
\draw[line width = 0.3mm] ({sin(3*360/5)+\ra}, {-cos(3*360/5)}) -- ({0+\ra}, {-1});


\draw[line width = 0.3mm] (0+\ra,-1) -- (0.8+\ra, -1.2);
\draw[line width = 0.3mm] (0+\ra,-1) -- (-0.7+\ra, -1.3);

\draw[line width = 0.3mm] ({sin(2*360/5)+\ra}, {-cos(2*360/5)}) -- (1.1+\ra, 0.6);
\draw[line width = 0.3mm] ({sin(2*360/5)+\ra}, {-cos(2*360/5)}) -- (1.2+\ra, 1.3);
\draw[line width = 0.3mm] ({sin(2*360/5)+\ra}, {-cos(2*360/5)}) -- (0.4+\ra, 1.6);
\draw[line width = 0.3mm] ({sin(3*360/5)+\ra}, {-cos(3*360/5)}) -- (-1.1+\ra, 0.4);
\draw[line width = 0.3mm] ({sin(3*360/5)+\ra}, {-cos(3*360/5)}) -- (-1.2+\ra, 1.1);

\draw[line width = 0.3mm] (0+\ra,-1) -- (0+\ra, -0);
\draw[line width = 0.3mm] ({sin(2*360/5)+\ra}, {-cos(2*360/5)}) -- (0+\ra, -0);
\draw[line width = 0.3mm] ({sin(3*360/5)+\ra}, {-cos(3*360/5)}) -- (0+\ra, -0);

\end{tikzpicture}\vspace{1mm}
\caption{Example of a local complementation on a graph. The local complementation is performed on the encircled vertex. The unconnected edges indicate that the graph shown can be part of a larger graph, that is left unchanged after the local complementation.}
\label{fig:lcexample}
\end{figure}

Finally, the \emph{chromatic index} of a graph $G$ will be useful for us to express minimum circuit depths with. The chromatic index of a graph $G$ is the smallest number of colors needed to color the edges of $G$ such that no two incident edges have the same color. 

\section{Distillation and error correction}\label{sec:correspondence}
In this section we define bilocal Clifford distillation protocols and stabilizer codes, and demonstrate a useful correspondence between the two.

\subsection{Bilocal Clifford protocols}
Bilocal Clifford protocols are distillation protocols where Alice and Bob first apply $C^\transp\otimes C^{\dagger}$, for some Clifford circuit $C$, see Fig.~\ref{fig:pingpong}. These Clifford circuits are composed of Hadamard gates $H$, $S$ gates, and $\CNOT$ gates. Afterwards, they measure out the last $n-k$ qubit pairs in the computational basis, and communicate their outcomes to each other. They both calculate the syndrome string $b$ of length $n$, where $b_i$ equals zero for $1\leq i\leq k$, and equals the parity of the sum of the two outcome bits of the measurement on the $i$'th pair for $k< i\leq n$. Depending on the outcome, Alice and Bob call the distillation a success or failure, and are otherwise allowed a final local unitary in the case of success. We will consider first only the case of post-selecting on $b= 0$ (which we will also refer to as the \emph{trivial measurement syndrome}), and consider the general case later in Section~\ref{sec:simpcalc}.

\begin{figure}[h!]
\centerfloat

\begin{tikzpicture}
\node[scale=1.25] at (-1.25,0.7){a)};
\node[scale=1.5] at (-1.1,0){$\bigotimes_{i=1}^{n}\rho_i$};
\node[scale=1.5] at (5.3,1/2){$A$};
\node[scale=1.5] at (5.3,-1/2){$B$};
\draw[line width=1.25] (0,0) -- (1,1/2);
\draw[line width=1.25] (5,1/2) -- (1,1/2);
\draw[line width=1.25] (0,0) -- (1,-1/2);
\draw[line width=1.25] (5,-1/2) -- (1,-1/2);
\node[draw, scale=1.25, fill=white] at (4.1,1/2){$C^{T}$};
\node[draw, scale=1.25, fill=white] at (4.1,-1/2){$C^{\dagger}$};
\node[anchor=east] at (5.8, 0) (node1){};
\node[anchor=east] at (5.8, -1.25) (node2){};
\draw[->, line width=0.3mm] (node1) to [out = -45, in = 45, looseness = 1] (node2);
\end{tikzpicture}
\vspace*{-5mm}

\begin{tikzpicture}
\node[scale=1.25] at (-1.25,0.7){b)};
\node[scale=1.5] at (-1,0){$\ket{\Phi^+}^{\otimes n}$};
\node[scale=1.5] at (5.3,1/2){$A$};
\node[scale=1.5] at (5.3,-1/2){$B$};
\draw[line width=1.25] (0,0) -- (1,1/2);
\draw[line width=1.25] (5,1/2) -- (1,1/2);
\draw[line width=1.25] (0,0) -- (1,-1/2);
\draw[line width=1.25] (5,-1/2) -- (1,-1/2);
\node[draw, scale=1.25, fill=white] at (3,-1/2){$\mathcal{N}_P$};
\node[draw, scale=1.25, fill=white] at (4.1,1/2){$C^T$};
\node[draw, scale=1.25, fill=white] at (4.1,-1/2){$C^{\dagger}$};
\node[anchor=east] at (5.8, 0) (node1){};
\node[anchor=east] at (5.8, -1.25) (node2){};
\draw[->, line width=0.3mm] (node1) to [out = -45, in = 45, looseness = 1] (node2);
\end{tikzpicture}
\vspace*{-5mm}

\begin{tikzpicture}
\node[scale=1.25] at (-1.25,0.7){c)};
\node[scale=1.5] at (-1,0){$\ket{\Phi^+}^{\otimes n}$};
\node[scale=1.5] at (5.3,1/2){$A$};
\node[scale=1.5] at (5.3,-1/2){$B$};
\draw[line width=1.25] (0,0) -- (1,1/2);
\draw[line width=1.25] (5,1/2) -- (1,1/2);
\draw[line width=1.25] (0,0) -- (1,-1/2);
\draw[line width=1.25] (5,-1/2) -- (1,-1/2);
\node[draw, scale=1.25, fill=white] at (3,-1/2){$\mathcal{N}_P$};
\node[draw, scale=1.25, fill=white] at (2,-1/2){$C$};
\node[draw, scale=1.25, fill=white] at (4.1,-1/2){$C^{\dagger}$};
\node[anchor=east] at (5.8, 0) (node1){};
\node[anchor=east] at (5.8, -1.25) (node2){};
\draw[->, line width=0.3mm] (node1) to [out = -45, in = 45, looseness = 1] (node2);
\end{tikzpicture}
\vspace*{-5mm}

\begin{tikzpicture}
\node[scale=1.25] at (-1.45,0.7){d)};
\node[scale=1.5] at (-1.25,0){$\ket{\Phi^+}^{\otimes n}$};
\node[scale=1.5] at (5.3,1/2){$A$};
\node[scale=1.5] at (5.3,-1/2){$B$};
\draw[line width=1.25] (0,0) -- (1,1/2);
\draw[line width=1.25] (5,1/2) -- (1,1/2);
\draw[line width=1.25] (0,0) -- (1,-1/2);
\draw[line width=1.25] (5,-1/2) -- (1,-1/2);
\node[draw, scale=1.25, fill=white] at (3,-1/2){$\mathcal{N}_{\tilde{P}}$};
\end{tikzpicture}
\caption{Depiction of how bilocal Clifford circuits map $n$-qubit-qubit pairs to $n$-qubit-qubit pairs before measuring. From a) to b), we use that $\otimes_{i=1}^n\rho_i = \left(I\otimes \mathcal{N}_P\right)\left(\left(\ket{\Phi^+}\bra{\Phi^+}\right)^{\otimes n}\right)$, with $\mathcal{N}_P\left(\cdot \right) = \sum_{P\in \mathcal{P}_n} p_{P} P\left(\cdot\right)  P^{\dagger}$.
In c), we use that $A^\transp\otimes I\ket{\Phi^+}^{\otimes n} = I\otimes A\ket{\Phi^+}^{\otimes n}$ for any matrix $A$ of the appropriate size~\cite{wilde2011classical}. For d), we use that Cliffords act on the group of Pauli strings $\mathcal{P}_n$ by conjugation. The channel can therefore be written as $\mathcal{N}_{\tilde{P}}\left(\cdot \right)= \sum_{P\in \mathcal{P}_n} p_{P} \tilde{P}\left(\cdot\right)  \tilde{P}^{\dagger}$ with $\tilde{P} = C^\dagger PC$.}
\label{fig:pingpong}
\end{figure}
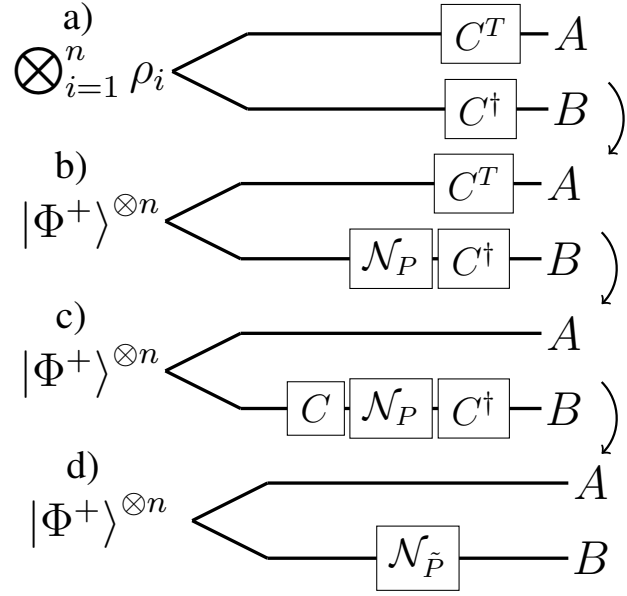

The states that Alice and Bob distill are Bell pairs $\ket{\Phi^+} = \frac{\ket{00}+\ket{11}}{\sqrt{2}}$ with noise applied to them. In particular, we assume Bell-diagonal noise, i.e.~$\mathcal{N}_P\left(\cdot \right) = \sum_{P\in \mathcal{P}_n} p_{P} P\left(\cdot\right)  P^{\dagger}$. That is, the noise corresponds to having applied the Pauli strings $P$ with probability $p_P$. We can assume without loss of generality that the noise is applied to only one side of the Bell pairs. This is due to the identity $A^\transp\otimes I\ket{\Phi^+}^{\otimes n} = I\otimes A\ket{\Phi^+}^{\otimes n}$, where $A$ is any matrix of the appropriate size~\cite{wilde2011classical}.
Bell-diagonal noise is not only a relevant error model~\cite{jansen2020enum}, but states can always be transformed to be of Bell-diagonal form by applying only local operations and classical communication whilst preserving the fidelity~\cite{bennett1996mixed}.

Define the set $\mathscr{P}_k$ by

\begin{equation*}
    \begin{split}
        \mathscr{P}_k = &\{P_1 \otimes\cdots \otimes~P_k  \otimes Q_{k+1} \otimes\cdots \otimes Q_{n} \in \mathcal{P}_n:\\
        &P_i \in \{I, X, Y, Z\}~\forall i \in \{1,\ldots,k\}, \\
        &Q_j \in \{I,Z\}\ \forall j \in \{k+1,\ldots,n\}\}.
    \end{split}
    \end{equation*}

The probability of a measurement with the all-zero syndrome string $b=0$ depends only on the set of $P\in\mathcal{P}_n$ that are mapped to $\mathscr{P}_k$ under the map $P \mapsto CPC^\dagger$~\cite{jansen2020enum}. Equivalently, these are all elements in the subgroup $C^\dagger\left(\mathscr{P}_k\right)$, where we abuse notation and use the shorthand $C^\dagger\left(\mathscr{P}_k\right) = \lbrace{C^\dagger P C \mid P \in \mathscr{P}_k \rbrace}$. 
The probability $p_\textrm{succ}^b$ for observing the $b=0$ syndrome is given by 

\begin{gather}
    p_\textrm{succ}^b = \sum_{\mathclap{P \in C^\dagger\left(\mathscr{P}_k\right)}} p_P\ .
    \label{eq:sucprob}
\end{gather}

Similarly, the fidelity for the all-zero syndrome string $b=0$ is determined by the $P\in C^\dagger\left(\mathscr{B}_k\right)$, where $\mathscr{B}_k$ is the set defined as

    \begin{equation*}
        \begin{split}
            \mathscr{B}_k &= \{I_1 \otimes \cdots \otimes I_k\otimes Q_{k+1} \otimes \cdots \otimes Q_{n}\} \in \mathcal{P}_n: \\
            &Q_j \in \{I,Z\}~\forall j \in \{k+1,\ldots,n\}\}.
        \end{split}
    \end{equation*}

The output fidelity $F^b$ (with respect to the $k$-fold tensor power of $\ket{\Phi^+}$) for the case of $b=0$ is given by

\begin{gather}
    F^b = \frac{\sum_{P \in C^\dagger\left(\mathscr{B}_k\right)}p_P}{\sum_{P \in C^\dagger\left(\mathscr{P}_k\right)}p_P} \ .
\end{gather}

As was shown in~\cite{jansen2020enum}, the set $C^\dagger\left(\mathscr{P}_k\right)$ determines the set $C^\dagger\left( \mathscr{B}_k\right)$ and vice versa. This is because the elements of $\mathscr{P}_k$ are exactly the elements that commute with all of $\mathscr{B}_k$, and vice versa. Since conjugation by Cliffords is an automorphism on $\mathcal{P}_n$, the image of $\mathscr{P}_k$ is uniquely determined by the image of  $\mathscr{B}_k$ (and vice versa) under such a conjugation. We note here that constructing the inverse of~$C$ (in particular in the symplectic picture) can be done efficiently. A distillation protocol is characterized by its \emph{distillation statistics} --- that is, the multiset of its output states (up to local operations) and success probabilities, for all possible values of $b$.

\subsection{Stabilizer codes}
A stabilizer group $\mathscr{B}$ is defined as an Abelian subgroup of the Pauli group on $n$ qubits $\mathcal{\thickbar P}_n$, not containing the $-I$ element. A stabilizer group acts on $\mathbb{C}^{2^n}$, the statespace of $n$ qubits, and stabilizes a subspace of dimension $2^k$. This subspace is the stabilizer code associated with $\mathscr{B}$. The \emph{basis codewords} of a stabilizer code are a (non-unique) collection of states that form a basis for the stabilized subspace, the elements of which we will also refer to as codewords.

Given a stabilizer group $\mathscr{B}$, let $\mathscr{B}^\perp$ be the set of elements in $\mathcal{\thickbar P}_n$ that commute with all elements in the stabilizer group. This set forms another group, which turns out to be an important group for quantum error correction~\cite{gottesman1997stabilizer}. In the symplectic picture, the two subgroups correspond to so-called isotropic and co-isotropic subspaces, respectively~\cite{de2011symplectic}, and form each others complement under the symplectic form $\omega$.

An important further quantity of a code is its \emph{distance}~$d$. The distance is the smallest weight error $E \in \mathcal{ P}_n$ that maps one codeword to another. In terms of the stabilizer group $\mathscr{B}$, this is the largest integer $d$ such that $\mathcal{E}_w\hspace{-0.5mm}\left(\mathscr{B}\right) = \mathcal{E}_w\hspace{-0.5mm}\left(\mathscr{B}^{\perp}\right)$, for all $0\leq w< d$, see~\cite{gottesman1997stabilizer}.

The Clifford group acts transitively on all stabilizer codes of fixed $n$ and $k$. In other words, given a fixed $[n, k, d]$ stabilizer code, it is possible to apply Clifford operations to it to obtain any other possible $[n, k, d']$ stabilizer code, which follows from the fact that the symplectic group acts transivitely on symplectic bases~\cite{de2011symplectic}.

For such a fixed stabilizer code, we can choose a particularly simple one. For given $n$ and $k$, we fix the stabilizer subgroup $\mathscr{B}_\textrm{base}$ as the one generated by $Z_{k+1}, Z_{k+2}, \ldots, Z_{n}$. Applying a Clifford circuit $C^\dagger$ to the stabilizer group $\mathscr{B}_\textrm{base}$ gives a new stabilizer group $C^\dagger \mathscr{B}_\textrm{base}C$. We have used $C^\dagger$ instead of $C$, which will turn out to be convenient later on. We note that the states stabilized by $\mathscr{B}_\textrm{base}$ are the states of the form $\ket{\psi}\ket{0}^{\otimes (n-k)} $, where $\ket{\psi}$ is an arbitrary state on $k$ qubits.
We note that \emph{stabilizer states} correspond precisely to $\left[n, 0, d\right]$ stabilizer codes~\cite{hein2006entanglement, hein2004multiparty}.

\subsection{Correspondence}
The above-mentioned stabilizer subgroup $\mathscr{B}_\textrm{base}$ is exactly the same as $\mathscr{B}_k$. Furthermore, $\mathscr{P}_k$ is the same as $\mathscr{B}_\textrm{base}^{\perp}$. Thus, applying $C^\dagger$ to $\mathscr{B}_\textrm{base}$ defines a new code $C^\dagger \mathscr{B}_\textrm{base}C$, which also sets the $P\in \mathcal{P}_n$ that get mapped $P\mapsto CPC^\dagger $ to $\mathscr{B}_k$. As mentioned above, this specifies the output state (up to local unitaries) and the success probability.
More explicitly, for a given stabilizer code that encodes a $k$-qubit state $\ket{\psi}$ into $n$ qubits by applying $C$ to $\ket{\psi}\ket{0}^{\otimes \left(n-k\right)}$, the corresponding distillation protocol corresponds to Alice and Bob applying the circuit $C^\transp\otimes C^\dagger$ and then measuring out the last $k$ states in the computational basis, in effect measuring the stabilizers of the code.

We show the correspondence in Fig.~\ref{fig:distandqec}. We note that the general case of the correspondence between quantum codes and distillation was considered in~\cite{aschauer2005quantum}, which we consider here a special case of, namely the correspondence between stabilizer codes and bilocal Clifford protocols. From now on, we will refer interchangeably to codes and distillation protocols.

\begin{figure}
\centerfloat
\def\dxx{-25mm}
\def\dyy{0.01mm}
\begin{tikzpicture}
\node[scale=1.4] at (3.8,1/2+0.6){$A$};
\node[scale=1.4] at (3.8,-1/2-0.6){$B$};
\draw[line width=0.75] (0,0) -- (1,1/2+0.5);
\draw[line width=0.75] (0,0) -- (1,1/2-0.4+0.5);
\draw[line width=0.75] (0,0) -- (1,1/2+0.4+0.5);
\draw[line width=0.75] (0,0) -- (1,1/2+0.8+0.5);

\draw[line width=0.75] (0,0) -- (1,-1/2-0.5);
\draw[line width=0.75] (0,0) -- (1,-1/2-0.4-0.5);
\draw[line width=0.75] (0,0) -- (1,-1/2+0.4-0.5);
\draw[line width=0.75] (0,0) -- (1,-1/2-0.8-0.5);

\draw[line width=0.75] (3,1/2+0.5) -- (1,1/2+0.5);
\draw[line width=0.75] (3,1/2-0.4+0.5) -- (1,1/2-0.4+0.5);
\draw[line width=0.75] (3,1/2+0.4+0.5) -- (1,1/2+0.4+0.5);
\draw[line width=0.75] (3,1/2+0.8+0.5) -- (1,1/2+0.8+0.5);

\draw[line width=0.75] (3,-1/2-0.5) -- (1,-1/2-0.5);
\draw[line width=0.75] (3,-1/2-0.5-0.4) -- (1,-1/2-0.4-0.5);
\draw[line width=0.75] (3,-1/2-0.5+0.4) -- (1,-1/2+0.4-0.5);
\draw[line width=0.75] (3,-1/2-0.5-0.8) -- (1,-1/2-0.8-0.5);



\node[draw, rounded rectangle, rounded rectangle west arc=none, scale=0.65, fill=white] at (3.1,1/2+0.5){$Z$};
\node[draw, rounded rectangle, rounded rectangle west arc=none, scale=0.65, fill=white] at (3.1,1/2-0.4+0.5){$Z$};

\node[draw, rounded rectangle, rounded rectangle west arc=none, scale=0.65, fill=white] at (3.1,-1/2-0.4-0.5){$Z$};
\node[draw, rounded rectangle, rounded rectangle west arc=none, scale=0.65, fill=white] at (3.1,-1/2-0.5-0.8){$Z$};

\draw[draw=black, fill=white] (1.5,1/2-0.3) rectangle ++(1.0,2.0) node[pos=.5, scale=1.4]{$C^\transp$};

\draw[draw=black, fill=white] (1.5,-1/2-0.3-1.4) rectangle ++(1.0,2.0) node[pos=.5, scale=1.4]{$C^\dagger$};

\node[anchor=east] at (3.5, 0) (node1){};
\node[anchor=east] at (3.5, -1.25) (node2){};




\draw[line width=0.75] (3-23mm-6mm,-1/2-0.5+1.65) -- (1-11.5mm,-1/2-0.5+1.65);
\draw[line width=0.75] (3-23mm-6mm,-1/2-0.5+1.65-0.4) -- (1-11.5mm,-1/2-0.5+1.65-0.4);
\draw[line width=0.75] (3-23mm-6mm,-1/2-0.5+1.65-0.8) -- (1-11.5mm,-1/2-0.5+1.65-0.8);
\draw[line width=0.75] (3-23mm-6mm,-1/2-0.5+1.65-1.2) -- (1-11.5mm,-1/2-0.5+1.65-1.2);



\node[fill=white, scale=1.22] at (0.3-34mm+2mm,-1/2-0.5+1*0.2+\dyy+14.0
){$\ket{\psi}$};



\node[fill=white, scale=0.79] at (0.3-34mm+2mm,-1/2-0.5-0.4+0.7+0.55){$\ket{0}$};
\node[fill=white, scale=0.79] at (0.3-34mm+2mm,-1/2-0.5-0.8+0.7+0.55){$\ket{0}$};




\draw[draw=black, fill=white] (1.5+\dxx,-1/2-0.3-1.4+\dyy-9mm) rectangle ++(1.0,2.0) node[pos=.5, scale=1.4]{$C$};
\draw [->, line width=1.8pt](3.7-9mm,-1.03+\dyy) -- (4.2-3.7mm,-1.03+\dyy);

\end{tikzpicture}
\vspace*{-2mm}
\caption{Relation between bilocal Clifford protocols and stabilizer codes, for the specific case of $n=4, k=2$. The left figure corresponds to a two-qubit state $\ket{\psi}$ being encoded into four qubits through a Clifford circuit $C$. The right figure shows a bilocal Clifford protocol, where $C$ is the same Clifford circuit as in the left. The circuit $C^\transp \otimes C^\dagger$ acts on the input state of the distillation protocol.}
\label{fig:distandqec}
\end{figure}
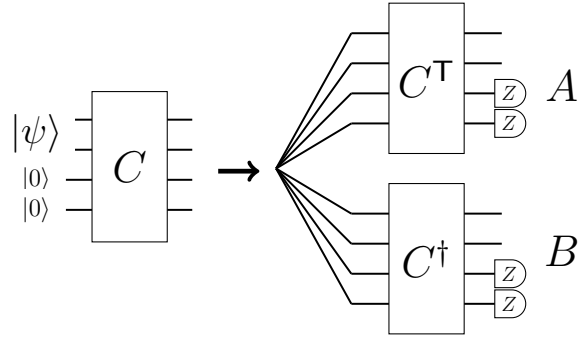

One detail here is that in the bilocal Clifford protocol picture a $\pm$ factor in front of a stabilizer is immaterial. In the stabilizer picture these prefactors do not change the actual error-correcting properties of the code, and we will ignore them here as well.

\section{Reduction to graph codes}\label{sec:graphreduc}
Here we show how we can reduce an optimization over all bilocal Clifford distillation protocols to one over a subset of \emph{graph codes} in the case of permutationally invariant \emph{depolarizing noise}. Depolarizing noise is a common noise model for quantum systems and for a single qubit corresponds to the following map $\rho \mapsto \left(1-p\right)\rho + p\textrm{Tr}\left(\rho\right)\frac{I}{2}$, where $\textrm{Tr}$ indicates the trace and $I$ is the identity operator on the corresponding qubit.

Graph codes are a subset of stabilizer codes, and any of the basis codewords can be conveniently described by a graph $G'$ with $n$ vertices, along with a linear combination of $k$ linearly independent bitstrings $\mathbf{a}_i$ of length $n$. First, we define

\begin{align}
    \ket{G'} = \prod_{\mathclap{\lbrace{i, j\rbrace} \in G'}}\;\CZ_{ij}\ket{+}^{\otimes n}\nonumber ,
\end{align}
where $\ket{+} = \frac{\ket{0}+\ket{1}}{\sqrt{2}}$, and where with $\prod_{\lbrace{i, j\rbrace} \in G}\CZ_{ij}$ we abuse notation to mean that a $\CZ_{ij}$ gate is applied for every edge $\lbrace{i, j\rbrace}$ in the graph $G'$. The set of basis codewords are then of the form

\begin{align}
Z^{\mathbf{b}}\ket{G'}\label{eq:codewords} \ ,
\end{align}
where $Z^{\mathbf{b}}$ is shorthand for a $Z$ gate for each qubit corresponding to a $1$ in the bitstring $\mathbf{b}\in \F_2^{n}$, and the $\mathbf{b}$ are all linear combinations of the $\mathbf{a}_i$. Since the $\mathbf{a}_i$ are linearly independent, there are $2^k$ distinct $\mathbf{b}$, so that the corresponding space is $2^k$-dimensional. The viewpoint of graph codes as built from a graph $G'$ with a collection of $Z$-type operators/bitstrings has been used in for example~\cite{yu2007graphical} to construct quantum error correction codes. We note that for the case of $k=0$, one retrieves the case of \emph{graph states}~\cite{hein2006entanglement, hein2004multiparty}, since the span of the empty set is the trivial vector space.

An $[n, k, d]$ graph code can also be described by the following procedure~\cite{schlingemann2001stabilizer, hein2006entanglement}, which will turn out to be useful for our purposes. First, prepare $n$ \emph{output qubits} in the $\ket{+}$ state, and prepare the state to be encoded in $k$ \emph{input qubits}. Now, $\CZ$ gates are applied between pairs of qubits, i.e.~$\prod_{\lbrace{i, j\rbrace} \in G}\CZ_{ij}$ for $G$ some graph is applied. Unlike the codeword picture, the graph $G$ here specifies the $\CZ$ gates to be applied also between input qubits and output qubits. As such, the graph $G$ has $n+k$ vertices, and not $n$ vertices as in the graph used in Eq.~\eqref{eq:codewords}.

To such a graph $G$ we can thus associate a (family of) states of the form $\prod_{\lbrace{i, j\rbrace} \in G}\CZ_{ij}\ket{\psi}\ket{+}^{\otimes n}$, where $\ket{\psi}$ is an arbitrary state on $k$ qubits. The choice of $\ket{\psi}$ only changes the state to be encoded, and does not change the error correcting properties of the code.

By measuring all the $k$ input qubits in the $X$ basis and applying a correction dependent only on the measurement outcomes, the input qubits are encoded in the $n$ remaining output qubits~\cite{hein2006entanglement}. To specify a graph code, it thus suffices to specify a graph $G$ and label the vertices as in- and output qubits, see Fig.~\ref{fig:graphexample} for an example. The example given there corresponds to the $\left[4, 2, 2\right]$ code~\cite{cafaro2014scheme}.

\begin{definition}
A graph is called an $\left(n, k\right)$-graph if its vertex set $V$ of size $n+k$ is partitioned into two sets $V^{\textrm{in}}$ and $V^{\textrm{out}} = V\setminus V^{\textrm{in}}$ of vertices (called the in- and output vertices), such that $\left|V^{\textrm{in}}\right| = k$.
\end{definition}

Furthermore, we will interchangably refer to input (output) qubits and input (output) vertices. Finally, we will refer to permutations of the vertices that permute the $n$ output and $k$ input vertices separately as $\left(n, k\right)$-permutations.

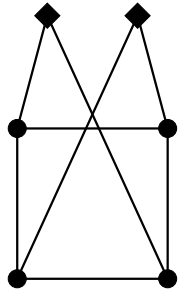
\begin{figure}[h!]
\centerfloat

\begin{tikzpicture}

\node[circle, fill=black, draw, scale=0.6] (a) at (-1,+1){};
\node[circle, fill=black, draw, scale=0.6] (b) at (+1, +1){};
\node[circle, fill=black, draw, scale=0.6] (c) at (+1, -1){};
\node[circle, fill=black, draw, scale=0.6] (d) at (-1, -1){};

\node[diamond, fill=black, draw, scale=0.6] (d1) at (0.6, 2.5){};
\node[diamond, fill=black, draw, scale=0.6] (d2) at (-0.6, 2.5){};


\draw[line width = 0.3mm] (1,1) -- (1,-1) -- (-1,-1) -- (-1,+1) -- cycle;

\draw[line width = 0.3mm] (-0.6, 2.5) -- (1,-1);
\draw[line width = 0.3mm] (-0.6, 2.5) -- (-1,1);

\draw[line width = 0.3mm] (0.6, 2.5) -- (-1,-1);
\draw[line width = 0.3mm] (0.6, 2.5) -- (1,1);




\end{tikzpicture}
\caption{Example of an $\left(n, k\right)$-graph, corresponding to the $[4, 2, 2]$ graph code~\cite{cafaro2014scheme}. The two-qubit input to the code is initialized on the diamond vertices, and then measured in the $X$ basis. After (local) corrections depending on the measurement outcomes, the input state is encoded on the remaining four vertices.}
\label{fig:graphexample}
\end{figure}

Let us now investigate the relation between the $\left(n, k\right)$-graph picture and the codeword picture from Eq.~\eqref{eq:codewords}. First let us consider the case of $k=1$, i.e.~a single input qubit. Fix an $\left(n, k\right)$-graph $G$ with a single input qubit (labelled by $v$), and prepare the input qubit in the state $\alpha\ket{0}+\beta\ket{1}$. A measurement on that input qubit leads (after a correction consisting solely of Pauli operations) to a state $\alpha \ket{G - \lbrace{v\rbrace}} + Z^{N_v}\beta \ket{G - \lbrace{v\rbrace}}$, where $\ket{G - \lbrace{v\rbrace}}$ is the graph state corresponding to the graph $G$ with vertex $v$ deleted, and $Z^{N_v}$ is shorthand for $\prod_{i\in N_v}Z_{i}$. Now let us consider $k$ arbitrary. After measuring out all input qubits $v_i \in V^{\textrm{in}}$ and applying the necessary corrections, we find that we end up with a superposition of (in general) $2^k$ states of the form
\begin{align}
    &Z^{\mathbf{b}}\ket{\left(G-V^{\textrm{in}}\right)}\ . \label{eq:codewords2}
\end{align}

By equating Eqs.~\eqref{eq:codewords} and \eqref{eq:codewords2}, we find that $G'=\left(G-V^{\textrm{in}}\right)$ is the graph obtained by removing all input vertices $v\in V^{\textrm{in}}$ from $G$, and that the possible $\mathbf{b}$ are linear combinations (over $\F_2$) of the $k$ strings $\mathbf{a}_i$. Importantly, the $\mathbf{a}_i$ are exactly those bitstrings that have a $1$ for the vertices in $G - V^{\textrm{in}} $ connected to $v_i$ for each $v_i\in V^{\textrm{in}}$, and zero otherwise. We note that the correction that needs to be performed is a stabilizer of $\ket{G'}$ (and thus consists of only Pauli corrections), and is chosen to anti-commute with exactly those $Z^{\mathbf{b}}$ that acquired a minus sign after the measurement.
While both the codeword (from Eq.~\eqref{eq:codewords}) and $\left(n, k\right)$-graph picture are useful for understanding graph codes, the $\left(n, k\right)$-graph picture will be more fruitful than the codeword picture for the enumeration of such codes. On the other hand, the codeword picture is particularly useful for understanding how to construct distillation circuits (see Section~\ref{sec:circuits}). For related literature on the $\left(n, k\right)$-graph picture, see~\cite{hwang2015relation, cafaro2014scheme}.

As mentioned above, graph codes are a strict subset of stabilizer codes that admit a convenient graphical representation. However, we will show that we can restrict to graph codes. First, let us define the subgroup $\mathcal{K}_n$ of the Clifford group on $n$ qubits as

$$\mathcal{K}_n=\langle\, \{\SWAP_{ij}\}_{1\leq i<j\leq n}\,\cup\,\{\Hgate_i\}_{i=1}^n\,\cup\,\{\Sgate_i\}_{i=1}^n\,\rangle.$$

This subgroup corresponds to permutation of the qubits, and single-qubit Clifford operations. We now define two equivalence relations on distillation protocols.

\begin{definition}\label{def:distillequiv}
Two bilocal Clifford distillation protocols are \emph{distillation equivalent} if the two protocols yield the same output states (up to local rotations) with the same success probability when distilling an $n$-fold tensor power of a Werner state and when conditioning on seeing the trivial measurement syndrome $b=0$.
\end{definition}

\begin{definition}\label{def:equiv1}
Two bilocal Clifford distillation protocols are \emph{locally equivalent} if their associated subgroups $\mathscr{B}_1$ and $\mathscr{B}_2$ are equal up to conjugation by an element $K$ in $\mathcal{K}_n$, i.e.~$$\mathscr{B}_1 \sim \mathscr{B}_2 \iff \mathscr{B}_1 = K \mathscr{B}_2 K^{-1}\ .$$
\end{definition}

The motivation for the first equivalence is clear --- if two protocols output the same state with the same probability, they are indistinguishable in their distillation capabilities, at least for $b=0$. Ideally, one would call two distillation protocols equivalent if for each syndrome string $b$ there exists another syndrome string~$b'$ such that the output state for the first protocol with syndrome string $b$ is the same as the output state up to local rotations for the second protocol with syndrome string $b'$. This is however impractical for enumeration purposes, since the number of possible syndrome strings grows as ~$2^{n-k}$, the number of coefficients to compare grows as $4^k$, and each coefficient is described by a weight enumerator of length $n+1$ (see Section~\ref{sec:simpcalc}). In Section~\ref{sec:results} we provide a heuristic motivation for restricting to the $b=0$ case. Thus, an enumeration over distillation protocols means finding a set of pairwise inequivalent distillation protocols for fixed $n$ and $k$. 
The second equivalence is motivated by the fact that $\mathcal{K}_n$ is the subgroup of the Clifford group that stabilizes an $n$-fold tensor power of a Werner state. Thus, the states before measuring when distilling with circuits $C$ and $CK$ with $K \in \mathcal{K}_n$ are equal, which means they are indistinguishable in their performance as a distillation circuit in the case of no noise. We note that the same equivalence was given in terms of double cosets in~\cite{jansen2020enum}, and that local equivalence implies distillation equivalence.

Now, every stabilizer code is equal to some graph code, up to single-qubit Cliffords~\cite{schlingemann2001stabilizer, grassl2002graphs}. This means that it suffices to consider graph codes up to permutation of the qubits. 

While every bilocal Clifford protocol is equivalent to a graph code, this graph code is not unique. This induces an equivalence relation on graph codes themselves. It will turn out to be most convenient to phrase this equivalence on $\left(n, k\right)$-graphs.

\begin{definition}\label{def:equiv2}
Two $\left(n, k\right)$-graphs $G_1,G_2$ are \emph{locally equivalent} if there are two stabilizer states $\ket{\psi}, \ket{\psi'}$ such that $\prod_{\lbrace{i, j\rbrace} \in G_1}\CZ_{ij}\ket{\psi}\ket{+}^{\otimes n}$ and $\prod_{\lbrace{i, j\rbrace} \in G_2}\CZ_{ij}\ket{\psi'}\ket{+}^{\otimes n}$ are the same up to (not necessarily single-qubit) Clifford operations on the input qubits and single-qubit Clifford operations plus permutations on the output qubits.
\end{definition}

This equivalence under single-qubit Clifford operations and permutations on the output qubits stems from the same reasoning as in definition~\ref{def:equiv1} when distilling Werner states. The equivalence under arbitrary Clifford operations on the input qubits stems from the fact that the state to be encoded does not change the error correcting properties of the code, as noted before. That is, the resultant codewords from Eq.~\eqref{eq:codewords} will not change, only their weights. The term locally equivalent is motivated by imagining the input qubits to being local to a single node, while the remaining qubits are assumed to be separated in space. We note that permutations on the output qubits are not local in this sense, however.

\begin{proposition}
Local equivalence on $\left(n, k\right)$-graphs is equivalent to the underlying $\left(n, k\right)$-graphs being related by a sequence of
$\left(n, k\right)$-permutations, local complementations and \emph{edge
flips}, i.e.~the addition or removal of an edge between two input vertices.
\end{proposition}
The above proposition follows from a result from~\cite{englbrecht2022transformations}, which deals with transforming graph states when qubits are grouped in such a way to be local to a node. In other words, each party is allowed to perform arbitrary Clifford operations on their locally held qubits. The result from~\cite{englbrecht2022transformations} now states that two graph states $\ket{G}, \ket{G'}$ are related by such \emph{party-local Clifford transformations} if and only if the underlying graphs are related by a sequence of edge flips and local complementations. Here, the edge flips are only allowed between vertices corresponding to a local party. 

Furthermore, the equivalence relation can be relaxed to a finer --- but better studied --- equivalence relation. 

\begin{corollary}\label{corr:enum1}
To enumerate all $[n, k, d]$ bilocal Clifford distillation protocols, it suffices to enumerate over all graphs with $n+k$ vertices up to graph isomorphism and local complementation, together with all subsets of the vertices with size $k$ (which effectively corresponds to selecting the $k$ input vertices).
\end{corollary}

We can furthermore restrict to connected graphs. That is because if $G$ is not connected, there are qubits that do not interact with each other. The corresponding distillation protocol would then naturally decompose into smaller distillation protocols. Connected representatives under the LC + permutation equivalence relation have been found up to $n=12$~\cite{cabello2011optimal}, meaning that in principle we can enumerate all $n$ to $k$ distillation protocols such that $n+k = 12$. We note that a restriction to connected graphs was not possible from the viewpoint considered in for example~\cite{yu2007graphical}.

For distillation protocols with $n+k>12$, a naive method would be to partition the set of $\left(n, k\right)$-graphs into the equivalence classes directly. Similar to the approach from \cite{danielsen2006classification, glynn2004geometry} a more efficient approach exists, however. This approach is based on so-called \emph{extensions}. We have not used this approach however, but detail it for completeness in Appendix~\ref{sec:extensions}.

We close this section with two remarks. First, a slightly more general scenario can be considered where besides in- and output qubits there exist also auxiliary qubits. Similarly to the output qubits, these qubits are prepared in the $\ket{+}$ state and have the $\CZ$ gates applied to them. Unlike the output qubits however, they are measured out in the $X$ basis, similar to the input qubits. Importantly, we do not have to consider the case of auxiliary qubits, since measuring an auxiliary qubit in the $X$ basis maps graph states to graph states, where importantly the two possible graph states that can arise are LC equivalent~\cite{hein2004multiparty}. Thus, the resulting states can be transformed by single-qubit Cliffords, and thus will yield equivalent codes.


Finally, we note that we restricted ourselves in definitions~\ref{def:equiv1} and \ref{def:equiv2} to equivalences phrased in terms of arbitrary Clifford operations, instead of arbitrary unitaries. This is motivated by the following. It was conjectured that equivalence of two graph states up to single-qubit unitaries implied equivalence up to single-qubit Clifford operations~\cite{van2005local, zeng2007local}. This was shown to be false, however~\cite{ji2007lu}. So far, there has been no good (graph-theoretical) understanding of the equivalence up to single-qubit unitaries for graph states, let alone for the case of $k>0$. For this reason, we consider only equivalence up to Clifford operations.

\section{Distillation circuits}\label{sec:circuits}
In the previous sections we used the $\left(n, k\right)$-graph representation to enumerate over bilocal Clifford distillation protocols. However, given an $\left(n, k\right)$-graph, it is not clear how to construct a bilocal Clifford circuit corresponding to the code. In particular, the encoding picture requires a total of $n+k$ qubits, while there exists a bilocal Clifford circuit that only processes $n$ qubits simultaneously.

In this section we provide first a way to construct a bilocal Clifford circuit from an $\left(n, k\right)$-graph. We then introduce heuristics for reducing the number of two-qubit gates (and/or optimize any other quantity of interest) of the corresponding circuits.

\subsection{From graph codes to circuits}
To find a circuit from a given graph code, we find a way to map the codewords of the code to codewords of the form
\begin{gather}
    Z^{\mathbf{b}'}\ket{+}^{\otimes n}\ ,
\end{gather}
where the $\mathbf{b}'$ are all the $2^k$ bitstrings that are $0$ on the last $n-k$ indices. These codewords are chosen since they correspond to the situation after decoding, see the left-hand side of Fig.~\ref{fig:distandqec}. 

The codewords of a graph code are always of the form shown in Eq.~\eqref{eq:codewords}. Applying the $\prod_{\lbrace{i, j\rbrace} \in (G-V^{\textrm{in}})}\CZ_{ij}$ circuit to such codewords yields codewords of the form $Z^{\mathbf{b}}\ket{+}^n$ (where we have assumed an ordering on the vertices). Since the $\mathbf{b}$ are the linear combinations of the $\mathbf{a}_i$, it suffices to map the $\mathbf{a}_i$ to a basis of the subspace that has a $0$ for all the qubits that are to be measured. Since the $\mathbf{a}_1, \mathbf{a}_2, \ldots, \mathbf{a}_k$ are linearly independent it is possible to bring the matrix $$\mathbf{A} = \left[\mathbf{a}_1, \mathbf{a}_2, \ldots, \mathbf{a}_k \right]^\transp$$ into row reduced echelon form with $k$ pivots. By relabeling the vertices, it is possible to set the reduced echelon form to have pivots in columns $1$ to $k$. It will be convenient to use such a labeling. In particular, let the in- and output vertices of an $\left(n, k\right)$-graph be labeled by $$V^{\textrm{in}} = \lbrace{v^{\textrm{in}}_i \rbrace}_{i=1}^{k} \textrm{ and }V^{\textrm{out}} = \lbrace{v^{\textrm{out}}_i \rbrace}_{i=1}^{n},$$ respectively. Such a labeling also splits the output vertices into those that are kept and measured out by setting $$V^{\textrm{out}}_{\textrm{keep}} = \lbrace{v^{\textrm{out}}_i \rbrace}_{i=1}^{k} \textrm{ and  } V^{\textrm{out}}_{\textrm{meas}} = \lbrace{v^{\textrm{out}}_i \rbrace}_{i=k+1}^{n},$$ respectively. 

\begin{definition}
A labeling $V^{\textrm{in}}$, $V^{\textrm{out}}$ is a \emph{valid labeling} if the row reduced echelon form of the matrix $\mathbf{A}$ has pivots in columns $1$ to $k$.
\end{definition}

An example of a valid labeling is shown in Fig.~\ref{fig:graph_to_circ}. A non-valid labeling would be one with output vertices $2$ and $4$ switched, since then $$\mathbf{A} = \begin{bmatrix} 1& 1 & 0 & 0\\ 0& 0 & 1 & 1 \end{bmatrix}$$ is already in reduced echelon form but has pivots in columns $1$ and $3$.

Given a valid labeling of an $\left(n, k\right)$-graph, it is possible to find a canonical set of $\CNOT$ gates (up to ordering) such that the $Z^{\mathbf{a}_{v_i}}$ operators are mapped to have support on only $V^{\textrm{out}}_{\textrm{keep}}$. In particular, for $1\leq i \leq k$, perform a $\CNOT_{ji}$ for every non-zero entry $j\neq i$ in the $i$'th row of~$\mathbf{A}$. For example, the matrix $$\mathbf{A} = \begin{bmatrix} 1& 0 & 1 & 1\\ 0& 1 & 0 & 1 \end{bmatrix}$$ corresponds to performing $\CNOT_{31}\CNOT_{41}\CNOT_{42}$. Note that the $\CNOT$ gates in this construction have the control on qubits in $V^{\textrm{out}}_{\textrm{meas}}$ and target on qubits in $V^{\textrm{out}}_{\textrm{keep}}$, and thus all commute. This fact will turn out to be useful for our heuristics for circuit construction later in this section.

Thus, to construct a circuit corresponding to an $\left(n, k\right)$-graph, a valid labeling needs to be established first. We emphasise that the labeling does not change the statistics when distilling an $n$-fold tensor power of a Werner state, and only affects the construction of the circuit. Then, $\CZ_{ij}$ is applied for each edge in the graph $G - V^{\textrm{in}} $. Afterwards the above construction for the $\CNOT$ gates is applied.
Finally, for each qubit in $V^{\textrm{out}}_{\textrm{meas}}$ a Hadamard is applied and then measured out. See Fig.~\ref{fig:graph_to_circ} for an example of the circuit constructed from an $\left(n, k\right)$-graph (with the associated valid labeling). We note a related approach was taken in~\cite{hwang2015relation}.

\begin{figure}[h!]
\centerfloat
\hspace{40mm}\begin{subfigure}{0.5\textwidth}

\begin{tikzpicture}

\node[circle, fill=black, draw, scale=0.6] (a) at (-1,+1){};
\node[circle, fill=black, draw, scale=0.6] (b) at (+1, +1){};
\node[circle, fill=black, draw, scale=0.6] (c) at (+1, -1){};
\node[circle, fill=black, draw, scale=0.6] (d) at (-1, -1){};

\node[diamond, fill=black, draw, scale=0.6] (d1) at (0.6, 2.5){};
\node[scale=1] (d) at (0.8, 2.8){$2$};
\node[diamond, fill=black, draw, scale=0.6] (d2) at (-0.6, 2.5){};
\node[scale=1] (d) at (-0.8, 2.8){$1$};

\node[scale=1] (d) at (-1-0.3, +1+0.3){$1$};
\node[scale=1] (d) at (+1+0.3, +1+0.3){$4$};
\node[scale=1] (d) at (-1-0.3, -1-0.3){$2$};
\node[scale=1] (d) at (+1+0.3, -1-0.3){$3$};

\draw[line width = 0.3mm] (1,1) -- (1,-1) -- (-1,-1) -- (-1,+1);

\draw[line width = 0.3mm] (-0.6, 2.5) -- (-1,1);
\draw[line width = 0.3mm] (-0.6, 2.5) -- (1,-1);
\draw[line width = 0.3mm] (-0.6, 2.5) -- (1,1);

\draw[line width = 0.3mm] (0.6, 2.5) -- (-1,-1);
\draw[line width = 0.3mm] (0.6, 2.5) -- (1,-1);


\draw [->, line width=1.8pt](3.7-1.3,+0.55) -- (4.2-1.3,0.55);
\end{tikzpicture}
\end{subfigure}\hspace{-35mm}%
\begin{subfigure}{0.5\textwidth}
\input{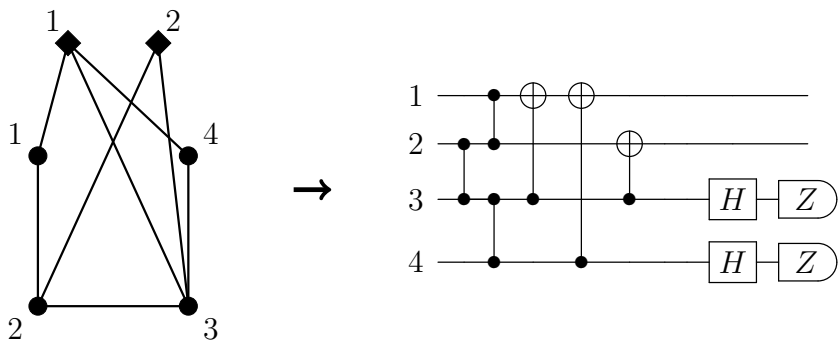}
\end{subfigure}
\caption{Constructing a circuit from an $\left(n, k\right)$-graph. The $\CZ$ gates correspond to the induced subgraph on the output qubits, while the $\CNOT$ gates map $Z_1Z_3Z_4$ (the neighbors of the left input qubit) and $Z_2Z_3$ (neighbors of the right input qubit) to $Z_1$ and $Z_2$.}
\label{fig:graph_to_circ}
\end{figure}

We now use this circuit picture to show that it is always possible to remove $\CZ$ gates that act only on qubits that are kept (i.e.~vertices in $V^{\textrm{out}}_{\textrm{keep}}$), without changing the distillation statistics. Important for the proof are the following commutation relations,

\begin{gather}
\CNOT_{ij}\CZ_{kl} = \CZ_{kl}\CNOT_{ij}\ \nonumber,\\
\CNOT_{ij}\CZ_{ij} = \CZ_{ij}\CNOT_{ij}\ \nonumber,\\ \CNOT_{ik}\CZ_{ij} = \CZ_{ij}\CNOT_{ik}\ \nonumber,\\
\CNOT_{ik}\CZ_{jk} = \CZ_{ij}\CZ_{jk}\CNOT_{ik}\ \label{eq:commutation},
\end{gather}
where the $i, j, k, l$ are distinct.

\begin{lemma}\label{lemma:removeczs}
Fix a valid labeling of an $\left(n, k\right)$-graph $G$. Then $G$ is locally equivalent to an $\left(n, k\right)$-graph $G'$ that has no edges between any pair of vertices in $V^{\textrm{out}}_{\textrm{keep}}$. Furthermore, the edge sets of $G$ and $G'$ differ only by edges $\lbrace{i, j\rbrace}$ with $i, j\in V^{\textrm{out}}_{\textrm{keep}}$ or $i, j\in V^{\textrm{out}}_{\textrm{meas}}$.
\end{lemma}
\begin{proof}
Given a valid labeling, there is a canonical circuit (up to ordering of the $\CZ$ and $\CNOT$ gates). Assume for that circuit there is some~$\CZ_{ij}$ with $i, j$ in~$V^{\textrm{out}}_{\textrm{keep}}$. Let us now attempt to commute this $\CZ_{ij}$ gate through one of the $\CNOT_{kl}$ gates, where by construction $k\in V^{\textrm{out}}_{\textrm{meas}}$ and $l\in V^{\textrm{out}}_{\textrm{keep}}$.
There are two cases --- either the $\CZ_{ij}$ and $\CNOT_{kl}$ gates commute, or a $\CZ_{ik}$ gate is added. Repeating this procedure until the $\CZ_{ij}$ gate is moved to the end will thus lead to a sequence of $\CZ_{pq}$ gates (with $p \in V^{\textrm{out}}_{\textrm{keep}},~q\in V^{\textrm{out}}_{\textrm{meas}}$) and $\CNOT_{kl}$ gates (where as before $k \in V^{\textrm{out}}_{\textrm{meas}},~l \in V^{\textrm{out}}_{\textrm{keep}}$), followed by the $\CZ_{ij}$ at the end. Note that the $\CZ_{ij}$ gate at the end is on qubits in~$V^{\textrm{out}}_{\textrm{keep}}$, and thus does not change the distillation statistics. Each of the other $\CZ_{pq}$  gates with $p \in V^{\textrm{out}}_{\textrm{keep}},~q\in V^{\textrm{out}}_{\textrm{meas}}$ can now be commuted back to the other $\CZ$ gates at the beginning of the circuit. As before, either a $\CZ_{pq}$ gate will commute with a $\CNOT_{kl}$ gate, or add a $\CZ_{qk}$ gate, with $q, k \in V^{\textrm{out}}_{\textrm{meas}}$. To summarize, since $q, k \in V^{\textrm{out}}_{\textrm{meas}}$, it is possible to commute a $\CZ$ gate acting on qubits in $V^{\textrm{out}}_{\textrm{keep}}$ through the $\CNOT$ gates (after which it can be ignored since it acts on qubit pairs that are to be kept), without introducing any $\CZ$ gate acting only on qubits in $V^{\textrm{out}}_{\textrm{keep}}$.

By repeating the above procedure for every $\CZ_{ij}$ gate with $i, j \in V_{\textrm{out}}$, there will eventually be no such $\CZ_{ij}$ gate remaining. Furthermore, this procedure only added~$\CZ$ gates between vertices in $V^{\textrm{out}}_{\textrm{meas}}$, and did not change any edges incident with $V^{\textrm{out}}$. Thus, since the above procedure did not depend on which valid labeling was used, the statement follows.
\end{proof}

We will use this Lemma in Section~\ref{sec:simpcalc} to find another way to enumerate distillation protocols using the symplectic formalism.

We close this section with the following two subtleties. While it is true that any $[n, k, d]$ code is locally equivalent to a graph code specified by an $\left(n, k\right)$-graph, the converse is not true. That is, while any $\left(n, k\right)$-graph specifies a stabilizer/graph code, it is not true that that code is necessarily a stabilizer $\left[n, k, d\right]$ code. A trivial example is given when none of the $k$ input qubits are connected with any output qubit. In this case, while the number of input qubits is greater than zero, the input state is prepared on a fixed state, and thus encodes no logical qubits. A less trivial example is given by the $\left(n, k\right)$-graph in Fig.~\ref{fig:counterexample}. Here, the problem is that the resultant codewords of the code span a space of dimension less than $k$. This is because the two input vertices share the same neighbors. As noted before, this is due to the fact that the $\mathbf{a}_i$ are not linearly independent. Note that such examples do not have any impact on any of the statements made in this section regarding our search for distillation protocols.

Finally, in the construction of the circuit a valid labeling of the $\left(n, k\right)$-graph was required. The labeling will lead to different constructed circuits, which could potentially lead to better circuits. We do not pursue optimizing over the different labelings, however.

\begin{figure}[h!]
\centerfloat

\begin{tikzpicture}

\node[circle, fill=black, draw, scale=0.6] (a) at (-1,+1){};
\node[circle, fill=black, draw, scale=0.6] (b) at (+1, +1){};
\node[circle, fill=black, draw, scale=0.6] (c) at (+1, -1){};
\node[circle, fill=black, draw, scale=0.6] (d) at (-1, -1){};



\node[diamond, fill=black, draw, scale=0.6] (d1) at (0.6, 2.5){};
\node[diamond, fill=black, draw, scale=0.6] (d2) at (-0.6, 2.5){};


\draw[line width = 0.3mm] (1,1) -- (1,-1) -- (-1,-1) -- (-1,+1) -- (1, -1);

\draw[line width = 0.3mm] (-0.6, 2.5) -- (-1,1);
\draw[line width = 0.3mm] (-0.6, 2.5) -- (1,1);

\draw[line width = 0.3mm] (0.6, 2.5) -- (-1,1);
\draw[line width = 0.3mm] (0.6, 2.5) -- (1,1);


\end{tikzpicture}
\caption{A graph with 2 input vertices, but whose corresponding code only encodes one qubit.}
\label{fig:counterexample}
\end{figure}
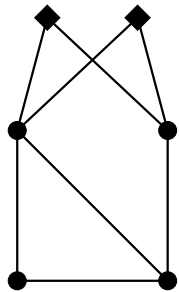

\subsection{Heuristics for circuit compilation}\label{sec:heuristics_for_circuit_compilation}
In the previous subsection we found a way to systematically construct a circuit from an $\left(n, k\right)$-graph. Here we are concerned with constructing \emph{good} circuits that achieve the same distillation statistics. Depending on the physical model, different criteria/metrics can be used for defining a good circuit. The first and most important metric we use is the number of two-qubit gates, which should be minimized. If decoherence over time is significant, it is important to minimize the depth of the circuit. If the gate noise is the predominant source of noise, we aim to reduce the number of two-qubit acting on the qubit(s) to be kept. We will refer to gates that act on the qubits to be kept as \emph{keep-gates} for short. In what follows, we detail three heuristics methods to search through a set of circuits that achieve the distillation statistics corresponding to a given $\left(n,k\right)$-graph. 

Firstly, given an $\left(n, k\right)$-graph $G$, we can construct a circuit using any $\left(n, k\right)$-graph that is locally equivalent to $G$. This is because the distillation statistics will necessarily be the same for the constructed circuits. As an example, we show a graph in Fig.~\ref{fig:graph_to_circ2} that is LC equivalent to the graph in Fig.~\ref{fig:graph_to_circ} (by an LC on vertex 2). Note that the graph in Fig.~\ref{fig:graph_to_circ} yields a shorter circuit.


\begin{figure}[h!]
\centerfloat
\hspace{40mm}
\begin{subfigure}{0.5\textwidth}

\begin{tikzpicture}

\node[circle, fill=black, draw, scale=0.6] (a) at (-1,+1){};
\node[circle, fill=black, draw, scale=0.6] (b) at (+1, +1){};
\node[circle, fill=black, draw, scale=0.6] (c) at (+1, -1){};
\node[circle, fill=black, draw, scale=0.6] (d) at (-1, -1){};

\node[diamond, fill=black, draw, scale=0.6] (d1) at (0.6, 2.5){};
\node[scale=1] (d) at (0.8, 2.8){$2$};
\node[diamond, fill=black, draw, scale=0.6] (d2) at (-0.6, 2.5){};
\node[scale=1] (d) at (-0.8, 2.8){$1$};

\node[scale=1] (d) at (-1-0.3, +1+0.3){$1$};
\node[scale=1] (d) at (+1+0.3, +1+0.3){$4$};
\node[scale=1] (d) at (-1-0.3, -1-0.3){$2$};
\node[scale=1] (d) at (+1+0.3, -1-0.3){$3$};

\draw[line width = 0.3mm] (1,1) -- (1,-1) -- (-1,-1) -- (-1,+1);

\draw[line width = 0.3mm] (-0.6, 2.5) -- (-1,1);
\draw[line width = 0.3mm] (-0.6, 2.5) -- (1,-1);
\draw[line width = 0.3mm] (-0.6, 2.5) -- (1,1);

\draw[line width = 0.3mm] (0.6, 2.5) -- (-1,-1);
\draw[line width = 0.3mm] (-1, 1) -- (1,-1);
\draw[line width = 0.3mm] (-1, 1) -- (0.6,2.5);
\draw [->, line width=1.8pt](3.7-1.3,+0.55) -- (4.2-1.3,0.55);
\end{tikzpicture}
\end{subfigure}
\hspace{-35mm}\begin{subfigure}{0.5\textwidth}
\input{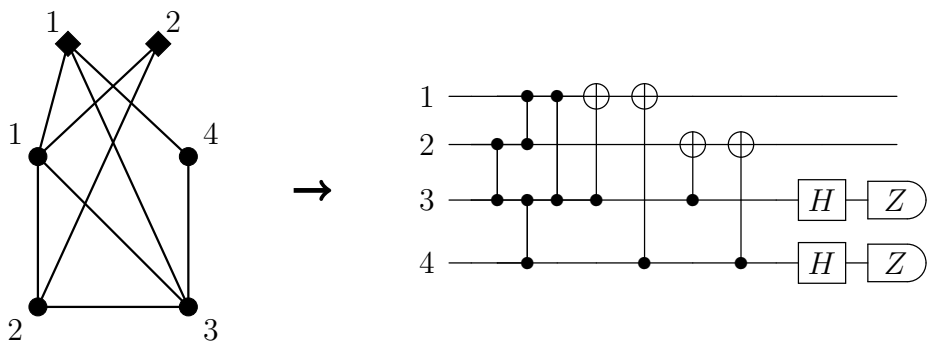}
\end{subfigure}
\caption{A graph code and corresponding circuit that is LC equivalent to the code and circuit in Fig.~\ref{fig:graph_to_circ}. Here the $\CNOT$ gates map $Z_1Z_3Z_4$ and $Z_1Z_2$ to $Z_1$ and $Z_1Z_2$, respectively.}
\label{fig:graph_to_circ2}
\end{figure}

Before moving on to the other heuristics, we investigate now briefly how to calculate (upper bounds) on the depth and the number of two-qubit gates corresponding to the circuit of an $\left(n, k\right)$-graph. First, the number of two-qubit gates is given by the sum of the number of $\CZ$ and $\CNOT$ gates. The number of $\CZ$ gates is equal to the number of edges in $G-V^{\textrm{in}}$. The number of $\CNOT$ gates is equal to the non-zero entries of the reduced row echelon form of $\left[\mathbf{a}_1, \ldots, \mathbf{a}_k\right]^\transp$ minus the number of pivots. The depth needed to perform the $\CZ$ gates is equal to the chromatic index of $G-V^{\textrm{in}}$, see Section~\ref{sec:prelim}. 
For calculating the depth of the $\CNOT$ gates, we note that all the $\CNOT$ gates commute. Thus, the minimum depth for the $\CNOT$ gates is the chromatic index of the graph with $n$ vertices and an edge between two vertices $v_i, v_j$ if there is a $\CNOT_{ij}$ gate.
Finally, one more time step is needed to perform the layer of Hadamard gates.
We note that in certain cases it is possible to perform some of the $\CZ$ and $\CNOT$ gates at the same time, which can reduce the depth even further.

Secondly, it is possible to change the order of all of the $\CZ$ and $\CNOT$ gates by commuting all of them through each other. For this, we use the commutation relations from Eq.~\eqref{eq:commutation}. In certain cases, the additional $\CZ$ gates incurred will cancel with $\CZ$ gates already present, leading potentially to a smaller number of two-qubit gates/depth/keep-gates.

In the above paragraphs we had circuits that first had a round of $\CZ/\CNOT$ gates, followed by a round of $\CNOT/\CZ$ gates. As our final heuristic, we break this structure to find better circuits. First, note it is possible to apply a $\CNOT$ gate (just before measuring) with control and target on the $n-k$ qubits that are measured out, without changing the distillation statistics. By commuting such a gate through (one of) the $\CZ$ gates, it is possible that some $\CZ$ gates will cancel. This can lead to keeping the total number of two-qubit gates the same (or even lower them), but allowing in certain cases to reduce the depth/keep-gates. Similarly, we also consider the case when permuting at most one of the $\CZ$ gates with the $\CNOT$ gates. We show an example of our heuristics in Fig.~\ref{fig:circ_heuristic}. In this example, we reconstruct the circuit also presented in~\cite{jansen2020enum}, but which was found using a brute-force method.

\begin{figure}[h!]
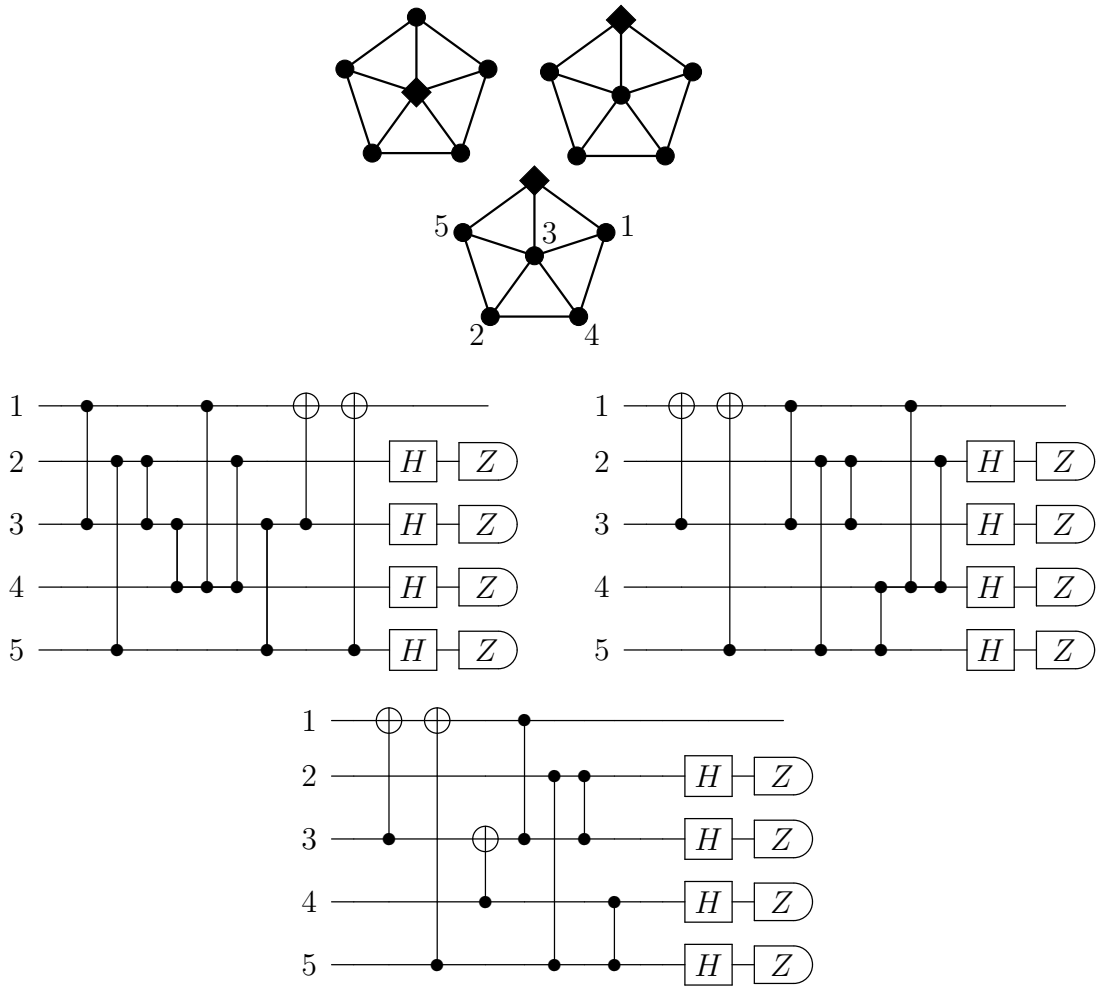

\centerfloat
\begin{subfigure}{0.15\textwidth}

\begin{tikzpicture}

\node[circle, fill=black, draw, scale=0.6] (1) at ({sin(0*360/5)}, {cos(0*360/5)}){};
\node[circle, fill=black, draw, scale=0.6] (2) at ({sin(1*360/5)}, {cos(1*360/5)}){};
\node[circle, fill=black, draw, scale=0.6] (3) at ({sin(2*360/5)}, {cos(2*360/5)}){};
\node[circle, fill=black, draw, scale=0.6] (4) at ({sin(3*360/5)}, {cos(3*360/5)}){};
\node[circle, fill=black, draw, scale=0.6] (5) at ({sin(4*360/5)}, {cos(4*360/5)}){};

\node[diamond, fill=black, draw, scale=0.7] (0) at (0.0, 0){};



\draw[line width = 0.3mm] (0,1) -- ({sin(1*360/5)}, {cos(1*360/5)}) -- ({sin(2*360/5)}, {cos(2*360/5)}) -- ({sin(3*360/5)}, {cos(3*360/5)}) -- ({sin(4*360/5)}, {cos(4*360/5)}) -- (0, 1);

\draw[line width = 0.3mm] (0,1) -- (0, 0);
\draw[line width = 0.3mm] ({sin(1*360/5)}, {cos(1*360/5)}) -- (0, 0);
\draw[line width = 0.3mm] ({sin(2*360/5)}, {cos(2*360/5)}) -- (0, 0);
\draw[line width = 0.3mm] ({sin(3*360/5)}, {cos(3*360/5)}) -- (0, 0);
\draw[line width = 0.3mm] ({sin(4*360/5)}, {cos(4*360/5)}) -- (0, 0);




\end{tikzpicture}
\end{subfigure}%
\begin{subfigure}{0.15\textwidth}

\begin{tikzpicture}

\node[diamond, fill=black, draw, scale=0.7] (1) at ({sin(0*360/5)}, {cos(0*360/5)}){};
\node[circle, fill=black, draw, scale=0.6] (2) at ({sin(1*360/5)}, {cos(1*360/5)}){};
\node[circle, fill=black, draw, scale=0.6] (3) at ({sin(2*360/5)}, {cos(2*360/5)}){};
\node[circle, fill=black, draw, scale=0.6] (4) at ({sin(3*360/5)}, {cos(3*360/5)}){};
\node[circle, fill=black, draw, scale=0.6] (5) at ({sin(4*360/5)}, {cos(4*360/5)}){};

\node[circle, fill=black, draw, scale=0.6] (0) at (0.0, 0){};



\draw[line width = 0.3mm] (0,1) -- ({sin(1*360/5)}, {cos(1*360/5)}) -- ({sin(2*360/5)}, {cos(2*360/5)}) -- ({sin(3*360/5)}, {cos(3*360/5)}) -- ({sin(4*360/5)}, {cos(4*360/5)}) -- (0, 1);

\draw[line width = 0.3mm] (0,1) -- (0, 0);
\draw[line width = 0.3mm] ({sin(1*360/5)}, {cos(1*360/5)}) -- (0, 0);
\draw[line width = 0.3mm] ({sin(2*360/5)}, {cos(2*360/5)}) -- (0, 0);
\draw[line width = 0.3mm] ({sin(3*360/5)}, {cos(3*360/5)}) -- (0, 0);
\draw[line width = 0.3mm] ({sin(4*360/5)}, {cos(4*360/5)}) -- (0, 0);




\end{tikzpicture}
\end{subfigure}


\begin{tikzpicture}

\node[diamond, fill=black, draw, scale=0.7] (1) at ({sin(0*360/5)}, {cos(0*360/5)}){};
\node[circle, fill=black, draw, scale=0.6] (2) at ({sin(1*360/5)}, {cos(1*360/5)}){};
\node[circle, fill=black, draw, scale=0.6] (3) at ({sin(2*360/5)}, {cos(2*360/5)}){};
\node[circle, fill=black, draw, scale=0.6] (4) at ({sin(3*360/5)}, {cos(3*360/5)}){};
\node[circle, fill=black, draw, scale=0.6] (5) at ({sin(4*360/5)}, {cos(4*360/5)}){};

\node[circle, fill=black, draw, scale=0.6] (0) at (0.0, 0){};

\node[scale=1] (1l) at (0.2, 0.3){$3$};
\node[scale=1] (2l) at ({sin(1*360/5)*1.3}, {cos(1*360/5)*1.3}){$1$};
\node[scale=1] (3l) at ({sin(2*360/5)*1.3}, {cos(2*360/5)*1.3}){$4$};
\node[scale=1] (4l) at ({sin(3*360/5)*1.3}, {cos(3*360/5)*1.3}){$2$};
\node[scale=1] (5l) at ({sin(4*360/5)*1.3}, {cos(4*360/5)*1.3}){$5$};


\draw[line width = 0.3mm] (0,1) -- ({sin(1*360/5)}, {cos(1*360/5)}) -- ({sin(2*360/5)}, {cos(2*360/5)}) -- ({sin(3*360/5)}, {cos(3*360/5)}) -- ({sin(4*360/5)}, {cos(4*360/5)}) -- (0, 1);

\draw[line width = 0.3mm] (0,1) -- (0, 0);
\draw[line width = 0.3mm] ({sin(1*360/5)}, {cos(1*360/5)}) -- (0, 0);
\draw[line width = 0.3mm] ({sin(2*360/5)}, {cos(2*360/5)}) -- (0, 0);
\draw[line width = 0.3mm] ({sin(3*360/5)}, {cos(3*360/5)}) -- (0, 0);
\draw[line width = 0.3mm] ({sin(4*360/5)}, {cos(4*360/5)}) -- (0, 0);




\end{tikzpicture}

\input{circ_heuristic0new_again.tex}\hspace{15mm}
\input{circ_heuristic1.tex}

\input{circ_heuristic2.tex}
 \caption{The top left graph is known to correspond to the five qubit code~\cite{schlingemann2001stabilizer}. The top right graph is obtained after a local complementation on two adjacent output vertices. For the next graph we choose a specific labelling of the input qubits. The first circuit is constructed from the graph above it. The second circuit is obtained from commuting the $\CNOT$ gates through the $\CZ$ gates. The last circuit results from adding a $\CNOT_{34}$ gate before measuring (which does not change the statistics), and commuting it through the $\CZ$ gates. Note that the last circuit has a lower depth than the one above it. We note that the last circuit is the same circuit as found in~\cite{jansen2020enum}, but which was found using a brute-force method.}
\label{fig:circ_heuristic}
\end{figure}

Thus, to heuristically find good enough circuit(s) for a given $\left(n, k\right)$-graph $G$, we first sample $\left(n, k\right)$-equivalent graphs by randomly applying local complementations (using the implementation from~\cite{dahlberg2020transform}) and edge flips. For each $\left(n, k\right)$-equivalent graph $G'$, we calculated the number of two-qubit gates for the circuit found directly from $G'$, and also from the circuit found from commuting all $\CZ$ gates through the $\CNOT$ gates. Out of these, only the circuits with the smallest number of two-qubit gates was kept. After having sampled through a sufficient number of $\left(n, k\right)$-graphs, the heuristics from the previous paragraph are applied to minimize either the depth or number of keep-gates.

\section{Enumerating protocols and calculating statistics in the symplectic picture}\label{sec:simpcalc}
With the ability to enumerate all bilocal Clifford protocols, we need a way to gauge the performance of a given distillation protocol.
The quantities of interest are the success probability (for a given observed syndrome $b$) and the coefficients of the output state (conditioned on observing $b$). These quantities will depend on the initial probability distribution of the input state $\lbrace p_P\rbrace_{P\in \mathcal{P}_n}$ and the given Clifford circuit $C$. Calculating these quantities in the density matrix formalism becomes unwieldy and impractical. Luckily, all of the necessary calculations can be phrased in the stabilizer/symplectic formalism.

In this section, we first construct the symplectic matrix given an $\left(n, k\right)$-graph. Then, we show how to reduce the search space of distillation protocols to symplectic matrices of a certain form. We close with discussing how to calculate the quantities of interest for distillation.

\subsection{Constructing symplectic matrices}
Here we describe how to find the symplectic matrix $M$ given an $\left(n, k\right)$-graph. Following the recipe from Section~\ref{sec:circuits}, we first apply a $\CZ_{ij}$ gate for each edge $\lbrace{i, j\rbrace} \in G-V^{\textrm{in}}$. We use the fact that the symplectic representation of $\prod_{\lbrace{ i, j\rbrace} \in G-V^{\textrm{in}}}\CZ_{ij}$ is equal to $$\begin{bmatrix}I_n&0\\ \textrm{Adj}\left(G-V^{\textrm{in}}\right)&I_n\end{bmatrix}, $$ where $\textrm{Adj}\left(G-V^{\textrm{in}}\right) = \begin{bmatrix}Q&R^\transp\\ R&S\end{bmatrix}, $ is the adjacency matrix of $G-V^{\textrm{in}}$ where we have rewritten the matrix without loss of generality with $Q\in \F_2^{k\times k},~S\in \F_2^{\left(n-k\right)\times \left(n-k\right)}$ symmetric and $R \in \F_2^{\left(n-k\right)\times k}$.

Now, the $\CNOT_{ij}$ gates are applied. Let $T \in \F_2^{\left(n-k\right)\times k}$ be the matrix with $T_{i, j} = 1$ if a $\CNOT$ gate is performed between $j+k$ and $i$ and $0$ otherwise. Note that $T$ is the bottom $\left(n-k\right)\times k$ submatrix of $\mathbf{A}^\transp$. The resulting symplectic matrix is then of the form $$\begin{bmatrix}A^\transp&0\\B'&A\end{bmatrix},$$ where
\begin{align}
    A=\left[\begin{array}{c|c}I_{k}&0\\\hline &  \\T&\:I_{n-k}\:\\& \end{array}\right],\\\  B'=\left[\begin{array}{c|c}Q&R^\transp\\\hline & \\R+TQ&S\!+\!TR^\transp\!\\ &\end{array}\right]\ \label{eq:sympreduction} .
    \end{align}

Now the final layer of Hadamard gates is applied. For convenience, we multiply both from the left and right with $\Hgate^{\otimes n}$. Note that multiplying by the right with $\Hgate^{\otimes n}$ does not change the distillation statistics, since \mbox{$\Hgate^{\otimes n}\in \mathcal{K}_n$}.

The symplectic matrix is then of the form $$\begin{bmatrix}A&B'\\0&A^\transp\end{bmatrix},$$ with $A$ and $B'$ as above.

Note that by Lemma~\ref{lemma:removeczs} it suffices to consider those $\left(n, k\right)$-graphs such that $Q=0$. We then retrieve the following.

\begin{theorem}
\label{corr:dionslemmagen}
Given a symplectic matrix $M$ corresponding to a distillation protocol, there is always a matrix $M'$ of the following form that will yield the same distillation statistics,
$$
M'=\begin{bmatrix}A&B\\0&A^\transp\end{bmatrix},\  A=\left[\begin{array}{c|c}I_{m}&0\\\hline &  \\T&\:I_{n-k}\:\\& \end{array}\right],\  B=\left[\begin{array}{c|c}0&R^\transp\\\hline & \\R&S\!+\!TR^\transp\!\\ &\end{array}\right],$$
where $R\in \F_2^{\left(n-k\right)\times k}$ and $S\in \F_2^{\left(n-k\right)\times \left(n-k\right)}$ is symmetric with zeroes on the diagonal.
\end{theorem}

Now let $t_i,r_i$ be the $i$'th column of $T$ and $R$, respectively. Using a similar argument from~\cite{jansen2020enum}, it suffices to consider those $T,~R$ such that for each $1\leq i\leq k$ it holds that $t_i\leq r_i\leq t_i+r_i$. Furthermore, it suffices to consider for $S$ the adjacency matrices of all graphs of order $n-k$ up to graph isomorphism. This result is a generalization from Lemma V.I in~\cite{jansen2020enum}.

We use corollary~\ref{corr:enum1} and Theorem~\ref{corr:dionslemmagen} to perform our enumeration over distillation protocols. Interestingly, in certain cases one of the two approaches work better. For example, corollary~\ref{corr:enum1} allows for a full enumeration over all $n=9$ to $k=1$ protocols within a reasonable time, while this is not possible using the approach from Theorem~\ref{corr:dionslemmagen}. On the other hand, since a characterization of LC equivalent graphs is missing for up to $17$ vertices, we could only enumerate over all $n=10$ to $k=7$ protocols using Theorem~\ref{corr:dionslemmagen}.

\subsection{Distillation statistics from symplectic matrices}
With a given symplectic matrix $M$ in hand, we now turn to calculating the corresponding distillation statistics.
As defined before, let $b \in {\lbrace 0, 1\rbrace}^n$ be such that for $1\leq i\leq k$ $b_i = 0$, and $b_i$ is the parity of the two outcome bits of the measurement on the $i$'th pair for $k< i\leq n$.
Before delving into the calculations, let us first motivate the idea of post-selecting on sets of different measurement syndromes. A number of entanglement distillation protocols (such as those studied in~\cite{jansen2020enum} and~\cite{krastanov2019optimized}) were based on \emph{error detection} --- that is, only the $b=0$ case was deemed a success. On the other hand, one can consider all possible syndrome strings, such as done in~\cite{munro2015inside}. This is commonly called \emph{error correction}. Error detection succeeds with a lower probability than error correction (since there are less accepted syndromes), but will have a higher (average) fidelity. This motivates us to consider arbitrary sets of syndrome strings to accept --- this will lead to a more fine-grained trade-off between the success probability and average fidelity.

For the symplectic matrix $M$ corresponding to a given distillation protocol and observing a given syndrome $b$, we find a success probability of
\begin{gather}
    p_\textrm{suc}^b= \sum_{\mathclap{v \in M^{-1}\left(\mathscr{P}_k+v_b\right)}} p_v,
    \label{eq:sucprobmain}
\end{gather}

where $v_b$ is the symplectic representation of the operator $X^b = \prod_{i=1}^nX^{b_i}$. This is because observing the syndrome $b$ corresponds to applying the operator $X^b$ just before measuring. Furthermore, we abuse notation and use $\mathscr{P}_k$ to refer to the symplectic representation of $\mathscr{P}_k$. Similarly, the corresponding fidelity is
\begin{gather}
    F^b = \frac{\sum_{v \in M^{-1}\left(\mathscr{B}_k+v_b\right)}p_v}{\sum_{v \in M^{-1}\left(\mathscr{P}_k+v_b\right)}p_v}\ .
    \label{eq:fidelity1main}
\end{gather}

The fidelity corresponds to the coefficient belonging to the identity Pauli string. Generalizing the above, the coefficient $F_{P}^{b}$ corresponding to an arbitrary Pauli string~$P$ is 
\begin{gather}
    F_{P}^b = \frac{\sum_{v \in M^{-1}\left(\mathscr{B}_k+v_b+v_P\right)}p_v}{\sum_{v \in M^{-1}\left(\mathscr{P}_k+v_b\right)}p_v}\ .
    \label{eq:fidelity2main}
\end{gather}
where $v_P$ is the symplectic representation of $P$.

We will now specialize to simplifying the calculation for the case of distilling an $n$-fold tensor power of a Werner state. In the case of distilling an $n$-fold tensor power of a Werner state, the coefficient $p_P$ of a Pauli string $P$ is entirely determined by the input fidelity and the weight $\textrm{wt}\left(P\right)$ of the string. Concretely,
\begin{gather}
p_P = F^{n-\textrm{wt}\left(P\right)}\left(\frac{1-F}{3}\right)^{\textrm{wt}\left(P\right)}\ ,
\end{gather}

where $F$ is the initial fidelity of the input Werner states.

This implies that it is sufficient to keep track only of the number of different weight operators for calculations. In particular, in the terminology introduced in Section~\ref{sec:prelim}, it suffices to consider the weight enumerators $\mathcal{E}_w\hspace{-0.5mm}\left( M^{-1}\left(\mathscr{B}_k+v_b+v_P\right)\right)$ and $\mathcal{E}_w\hspace{-0.5mm}\left( M^{-1}\left(\mathscr{P}_k+v_b\right)\right)$, see Section~\ref{sec:prelim}. That is,
\begin{align}
    F_{P}^b =&~\frac{\sum_{v \in M^{-1}\left(\mathscr{B}_k+v_b+v_P\right)}p_v}{\sum_{v \in M^{-1}\left(\mathscr{P}_k+v_b\right)}p_v}\nonumber \\
    =&~\frac{\sum_{w=0}^n\mathcal{E}_w\hspace{-0.5mm}\left( M^{-1}\left(\mathscr{B}_k+v_b+v_P\right)\right)F^{n-w}\left(\frac{1-F}{3}\right)^{w}}{\sum_{w=0}^n\mathcal{E}_w\hspace{-0.5mm}\left( M^{-1}\left(\mathscr{P}_k+v_b\right)\right)F^{n-w}\left(\frac{1-F}{3}\right)^{w}}\nonumber \\
    =&~\frac{\mathcal{E}\hspace{-0.5mm}\left( M^{-1}\left(\mathscr{B}_k+v_b+v_P\right), F, \frac{1-F}{3}\right)}{\mathcal{E}\hspace{-0.5mm}\left( M^{-1}\left(\mathscr{P}_k+v_b\right), F, \frac{1-F}{3}\right)}\ .\label{eq:weightexpression}
\end{align}

Similarly, we find that the success probability equals 
\begin{gather}
    p_\textrm{suc}^b= \mathcal{E}\hspace{-0.5mm}\left( M^{-1}\left(\mathscr{P}_k+v_b\right), F, \frac{1-F}{3}\right)\label{eq:sucprobnew}\ .
\end{gather}

Furthermore, we do not have to find the individual summands of the numerator and denominator of Eq.~\ref{eq:weightexpression} for the case of $b=0,~P=I^{\otimes n}$. This is because     $\mathcal{E}\hspace{-0.5mm}\left(M^{-1}\left(\mathscr{P}_k\right)\right)_w$ and     $\mathcal{E}\hspace{-0.5mm}\left(M^{-1}\left(\mathscr{B}_k\right)\right)_w$ are related by the so-called quantum MacWilliams identity~\cite{shor1997quantum, gottesman1997stabilizer},
\begin{align}
\centerfloat
    &~2^{n-k}\cdot \hspace{0.2mm}\mathcal{E}_w\hspace{-0.5mm}\left(M^{-1}\left(\mathscr{P}_k\right)\right)\label{eq:macwilliams}\  \\
    =&~\sum_{w'=0}^n\left[  \sum_{s=0}^{w}(-1)^s3^{w-s} \binom{w'}{s}\binom{n-w'}{w-s} \right]\mathcal{E}_{w'}\hspace{-0.5mm}\left(M^{-1}\left(\mathscr{B}_k\right)\right)\nonumber \ .
\end{align}

Calculating the probability using Eq.~\ref{eq:sucprobnew} requires $2^{n+k}$ sums. However, using Eq.~\ref{eq:macwilliams} it suffices to calculate only $\mathcal{E}\hspace{-0.5mm}\left(M^{-1}\left(\mathscr{B}_k\right)\right)_w$, which requires only a sum over $2^{n-k}$ terms, and then performing $\mathcal{O}\left(n^3\right)$ sums. This gives a speedup for calculating the fidelity and success probability for the case of $b=0$.

This motivates generalizing the MacWilliams identity to the case of $b\neq 0$. That is, finding a relationship between $$\mathcal{E}_{w}\hspace{-0.5mm}\left(M_1^{-1}\left(\mathscr{B}_k + v_b\right)\right) \textrm{ and } \mathcal{E}_{w}\hspace{-0.5mm}\left(M_2^{-1}\left(\mathscr{P}_k+v_P\right)\right)$$

We note that an invertible relation does not exist for the case of $v_b$ replaced with general $v_P$. This is because examples were found of symplectic matrices $M_1$ and $M_2$ such that $$\mathcal{E}_{w}\hspace{-0.5mm}\left(M_1^{-1}\left(\mathscr{B}_k\right)\right) = \mathcal{E}_{w}\hspace{-0.5mm}\left(M_2^{-1}\left(\mathscr{B}_k\right)\right).$$
but there exist no $P_1, P_2 \neq I^{\otimes n}$ such that $$\mathcal{E}_{w}\hspace{-0.5mm}\left(M_1^{-1}\left(\mathscr{B}_k+v_{P_1}\right)\right) = \mathcal{E}_{w}\hspace{-0.5mm}\left(M_2^{-1}\left(\mathscr{B}_k+v_{P_2}\right)\right).$$

More informally, this is because we found examples of symplectic matrices $M_1$ and $M_2$ such that the resulting states have the same fidelity and success probability, but the other coefficients of the output state differ (even after local operations). This is related to the existence of codes/stabilizer states that are locally inequivalent, yet share the same $\mathcal{E}_{w}\hspace{-0.5mm}\left(M^{-1}\left(\mathscr{B}_k\right)\right)$ and $ \mathcal{E}_{w}\hspace{-0.5mm}\left(M^{-1}\left(\mathscr{P}_k\right)\right)$~\cite{danielsen2006classification}.

We note here that, since an $[n, k, d]$ code has $\mathcal{E}\hspace{-0.5mm}\left(M^{-1}\left(\mathscr{B}_k\right)\right)_w = \mathcal{E}\hspace{-0.5mm}\left(M^{-1}\left(\mathscr{P}_k\right)\right)_w$ for all $w < d$~\cite{gottesman1997stabilizer}, expanding the expression for the fidelity for $b=0$ in Eq.~\eqref{eq:weightexpression} around $F_\textrm{in}=1$ gives a distillation protocol with output fidelity
\begin{align}
    F_\textrm{out} = 1- \frac{B_d-A_d}{3^d}(1-F_\textrm{in})^d +\mathcal{O}(1-F_\textrm{in})^{d+1})\ .
\end{align}

Finally, we note that it is also possible to formulate the calculation of the weight enumerators in terms of the $\left(n, k\right)$-graph only (i.e.~without constructing a symplectic matrix first). This is done by first constructing the codewords, and then calculating the weight distributions as in~\cite{yu2007graphical}. A related approach was given in~\cite{yu2007graphical}, where a graph-theoretical approach was given to calculate the distance of a graph code~\footnote{We note that in~\cite{yu2007graphical} the results were framed in terms of the graph codewords (see Eq.~\eqref{eq:codewords}) and not $\left(n, k\right)$-graphs. However as noted in this paper, these different approaches can be mapped to one another.}.

\section{Results}\label{sec:results}
We have used our tools to find practical distillation protocols, which we now report on here. As in the previous sections, we focus on the scenario of distilling an $n$-fold tensor power of a Werner state. 

First, we investigate the potential benefits that considering non-trivial measurement syndromes (i.e.~$b\neq 0$) can give for $n$ to $1$ distillation. Secondly, we evaluate how well the heuristically found circuits perform under gate- and measurement noise. We compare the output fidelities of our circuits with those found using the genetic algorithm from~\cite{krastanov2019optimized}. Finally, we explore the advantages more general $n$ to $k$ distillation protocols can bring in comparison with $n$ to $1$ distillation. To this end, we use the highest fidelity $10$ to $7$ distillation protocol to teleport one half of a maximally entangled state encoded in the Steane code between two parties. We compare this approach with two more standard approaches --- one based on no distillation at all, and one that concatenates the $2$ to $1$ DEJMPS distillation protocol~\cite{Deutsch1996}.

\subsection{Non-trivial measurement syndromes}
For our first exploration of the impact of non-trivial measurement syndromes, we consider both the success probability and output fidelity $F_\textrm{out}$ for different input fidelities $F_{\textrm{in}}$. In Fig.~\ref{fig:fidelity2} we consider the envelope of all found protocols, both with only $b=0$ (solid) and optimizing over all syndrome sets (dashed). Since the possible number of syndrome sets to condition is equal to $2^{2^{n-k}}$, the results shown are only for up to $n=5$.

\begin{figure}
\centerfloat
	\includegraphics[clip,  width=0.5\textwidth, trim = 1.5mm 1.5mm 3.1mm 0mm]{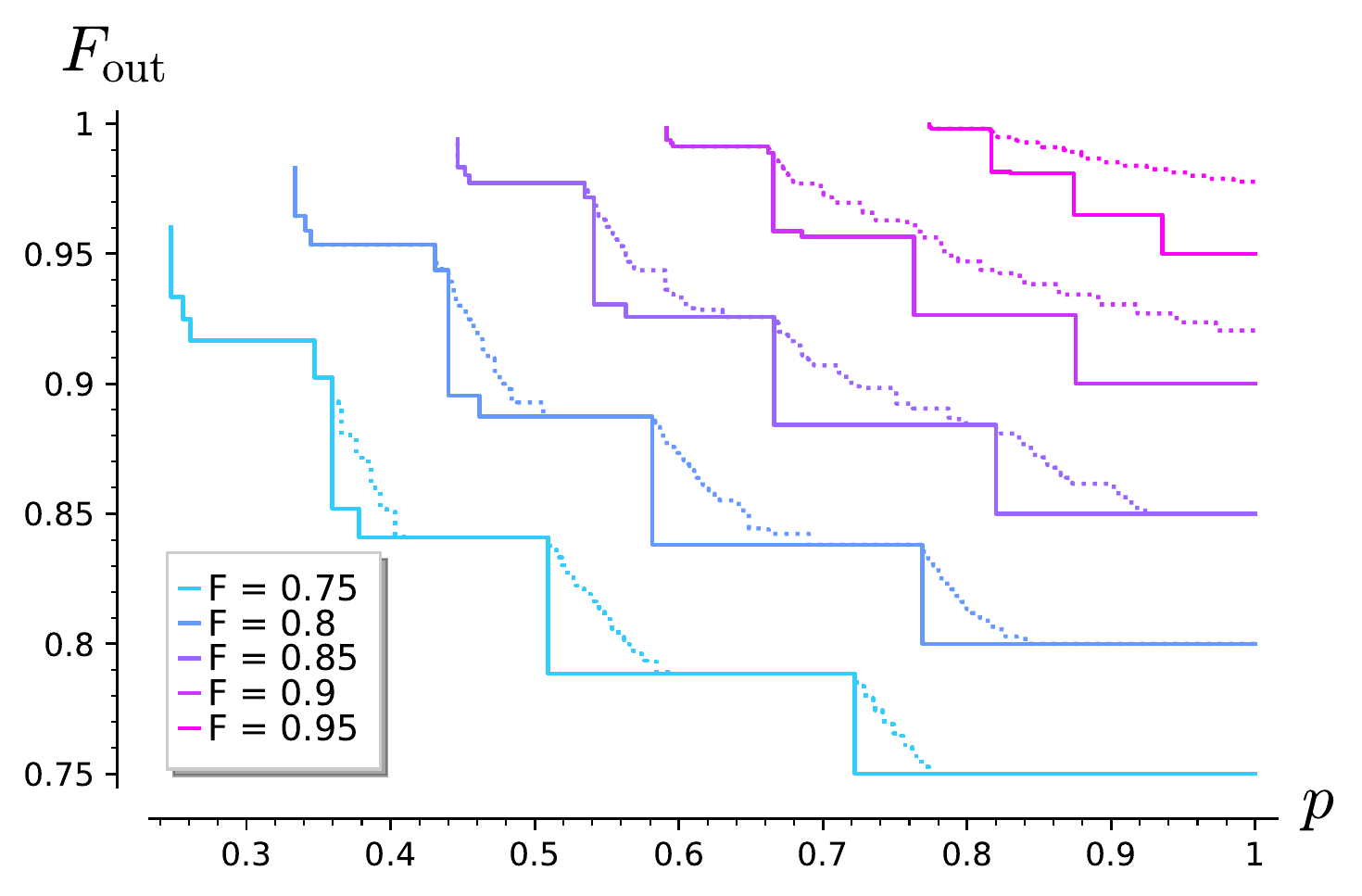}
	\vspace*{-1.8mm}
	\caption{Envelope of the achieved fidelity and success probabilities, where we compare post-selecting on detecting all correlated outcomes (solid) with optimizing over all possible corrections after measuring (dotted). The envelope is shown for several input fidelities, and for all $n$ to $1$ bilocal Clifford protocols with $n=2\ldots 5$.}
	\label{fig:fidelity2}
\end{figure}

From Fig.~\ref{fig:fidelity2} it can be seen that including non-trivial measurement syndromes provides a more significant benefit for larger input fidelities. However, note that it is in principle possible to always achieve the convex hull of a set of distillation protocols by probabilistically mixing distillation protocols. Observe that the convex hull of the solid and dotted lines are equal for input fidelities equal to or less than $0.85$.
This implies that for input fidelities $\lesssim 0.85$ the inclusion of non-trivial measurement syndromes provides no benefit, while for input fidelities somewhere in between $0.85$ and $0.95$ non-trivial measurement syndromes start to perform better than probabilistic mixing of trivial measurement syndromes. This is consistent with the results from~\cite{munro2015inside}.

Secondly, we consider using distillation for quantum key distribution. We consider the secret-key rate achieved when using asymptotic asymmetric BB84~\cite{bennett2020quantum} after performing $n$ to 1 distillation. Furthermore, we consider two different approaches. Firstly, we consider only using the output state when measuring a trivial measurement syndrome $b=0$. Secondly, we consider using all the possible states for each possible syndrome string $b$. Importantly, we \emph{bin} the states. That is, we separate the measured statistics into bins according to the syndrome string $b$. This allows us to separate the observed bits into those that had smaller or greater quantum bit error rates. From the convexity of the secret-key rate this can lead to increased secret-key rates, see for example~\cite{jing2020quantum} for a similar approach. We show the resultant rates for $n=2, \ldots, 7$ in Fig.~\ref{fig:qkd}, where the solid line corresponds to the above-mentioned binning approach, the dotted line corresponds to only using the syndrome string $b=0$. The plot only shows the results for up to $n=7$, since calculating the output states for the $2^{n-1}$ different syndromes became too computationally intensive.

\begin{figure}
\centerfloat
	\includegraphics[clip,  width=0.54\textwidth, trim = 0mm 0mm 0mm 0mm]{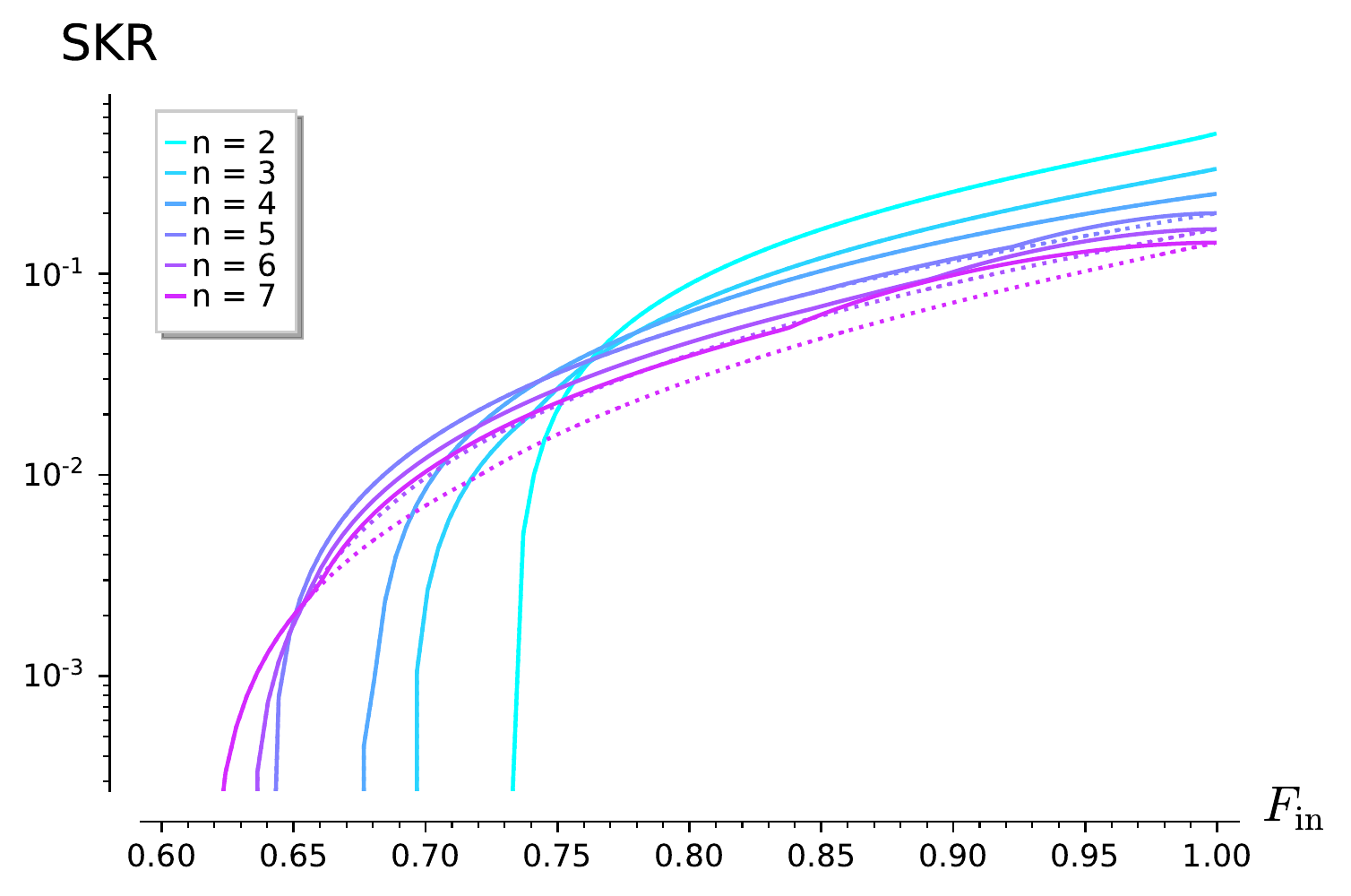}
	\vspace*{-0.2mm}
	\caption{Achieved secret-key rate using the asymptotic BB84 protocol after distilling from $n$ to 1 pairs, where the envelope is taken over all $n$ to 1 protocols for fixed $n$. The solid line corresponds to separating the generated states into bins according to the observed syndrome, and performing the BB84 post-processing for each such bin separately. The dotted line corresponds to only using the state with the syndrome string $b = 0$ (i.e.~error detection).}
	\label{fig:qkd}
\end{figure}

As would be expected, distilling with a larger number of pairs allows for a higher noise tolerance. Furthermore, we see that the envelope of both strategies is the same. This thus suggests that it suffices to condition only on the $b =0$ syndrome when one can choose the number of pairs $n$ to distill one pair out of, similar to the conclusion from~\cite{jing2020quantum}. Even for larger $n$, any potential difference between the strategies would be marginal and for a small range of fidelities.

That is, for tasks such as for example QKD, considering non-trivial measurement syndromes does not provide a benefit. This then provides a heuristic motivation for the equivalence defined in Definition~\ref{def:distillequiv}, where two distillation protocols were deemed distillation equivalent if the output states for $b=0$ were the same up to local rotations. On the other hand, deterministic distillation (i.e.~including all possible measurement syndromes) is a key component of second generation quantum repeaters~\cite{munro2015inside}. Furthermore, it is not clear how non-trivial syndromes would impact the capabilities of general $n$ to $k$ bilocal Clifford protocols, especially for such tasks as QKD.

We conclude this subsection by noting that a possible strategy is to take the average state over all syndrome strings $b$ after local corrections. However, for the values of $n$ considered here, this only increases the output fidelity for input fidelities $F_\textrm{in} \gtrapprox0.88$~\cite{munro2015inside}. Since asymptotic BB84 requires an input fidelity of $F_\textrm{in}\gtrapprox0.835$ (assuming a Werner state as input), distilling does not allow for generating key at input fidelities lower than $F_\textrm{in}\gtrapprox0.835$. At the same time, the fact that more states are used and the success probabilities drop down as $n$ increases, leads to the fact that distilling with bilocal Clifford protocols with such a strategy does not bring any benefits for quantum key distribution. This shows the benefits of using additional measurement information and binning accordingly for certain quantum communication tasks~\cite{jing2020quantum}.

\subsection{Noisy circuit comparison}\label{sec:noisy_circuit_comparison}
The results from the previous section assumed perfect gates and measurements. In practice operations will be noisy, reducing the benefits of distillation. This motivates us to investigate how well our found circuits perform in the case of noise. As a comparison, we use the genetic algorithm tools from~\cite{krastanov2019optimized}. The approach taken there is to represent purification protocols as sequences of gates, however, permitting only gates that map Bell states to other Bell states. As detailed in the Appendix, that is sufficient to describe the purification protocols considered here and it permits very efficient simulation. Moreover, the simulation can take into account local gate and measurement noise, not only network noise in the initial Bell pairs. Thus, the optimizer, which is a simple genetic algorithm over the sequence of gates, can find circuits more resilient to the imperfections of real hardware.

We note that the framework from~\cite{krastanov2019optimized} explicitly allows for the optimization of circuits in the case of there being a limit on the number of qubits that can be processed simultaneously. Such considerations are especially relevant for distillation on NISQ devices~\cite{krastanov2019optimized, preskill2018quantum}. In the framework considered in the present paper, there is no such restriction. Furthermore, the software from~\cite{krastanov2019optimized} allows for an optimization when considering arbitrary Pauli noise, i.e.~it is not restricted to depolarizing noise.

Lastly, the genetic algorithm black-box optimizer needs to be executed for every set of hardware parameters, as different levels of noise might be addressed by different circuits, as seen in Fig.~\ref{fig:noisycircuit4}.

\begin{figure*}
\centerfloat
	\includegraphics[clip,  width=\textwidth]{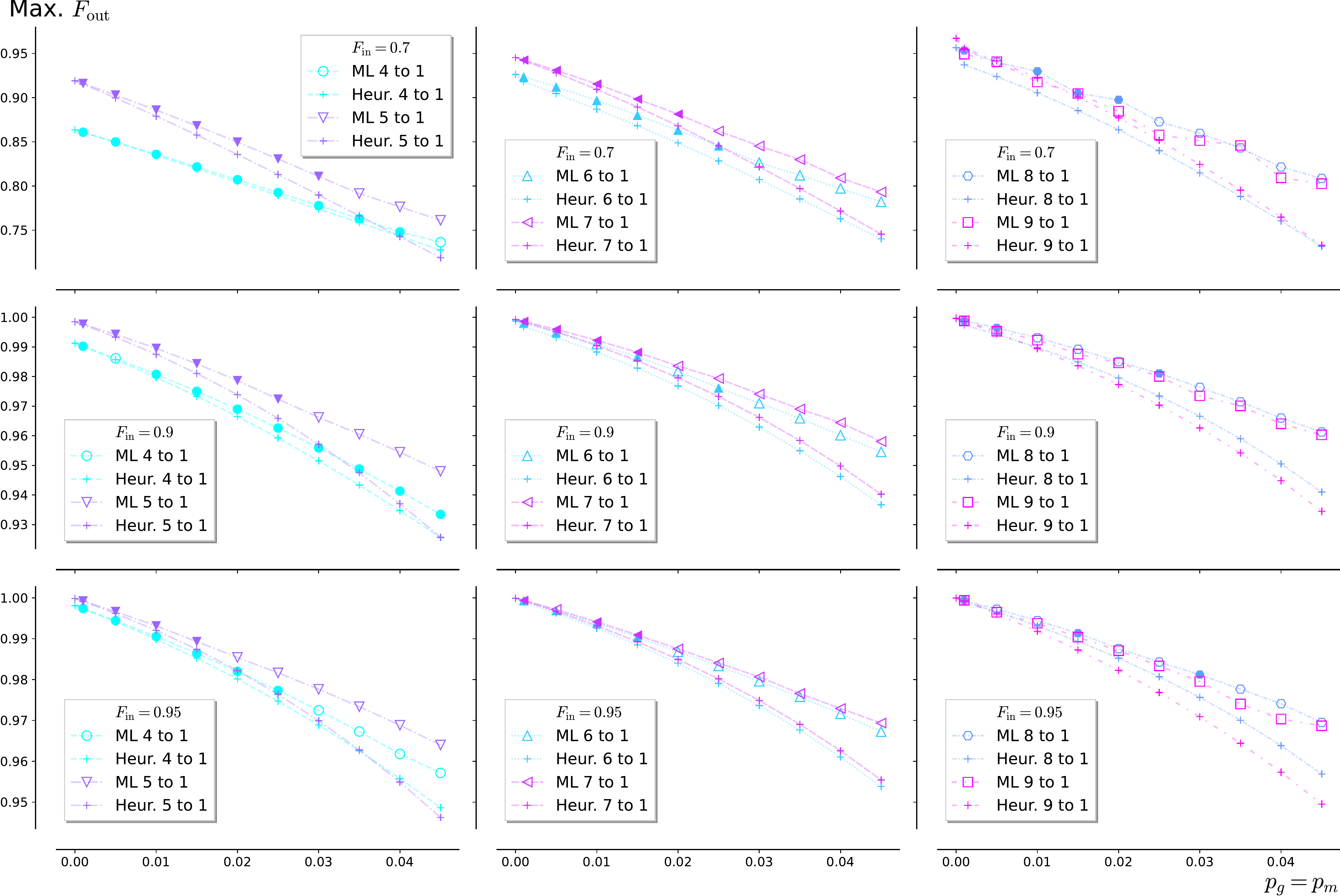}
	\vspace*{-3.2mm}
	\caption{Output fidelity when distilling from $n$ to $1$ pairs as a function of the gate noise. The input fidelity of the initial states is $F_\textrm{in}=0.7$ (top row), $F_\textrm{in}=0.9$ (middle row), and $F_\textrm{in}=0.95$ (bottom row). We consider two cases: distilling with circuits found using our heuristics (`Heur.') and distilling with circuits found using the software from~\cite{krastanov2019optimized} (`ML'). For the heuristic results, all $n$ to 1 data points for a specific $n$ represent the same circuit: for each $n$, this circuit can be found in Appendix \ref{sec:details_noisy_circuit_comparison}. For the `ML' data, each data point is a specific circuit that came out of the black-box optimizer. For these data, an open (closed) marker indicates that this circuit has different (the same) distillation statistics as the optimal circuit in case of no gate and measurement noise (\textit{i.e.}, as the corresponding circuit of Appendix \ref{sec:details_noisy_circuit_comparison}).}
	\label{fig:noisycircuit1}
\end{figure*}

We model the noise in the circuit by gate and measurement noise. Measurement noise is modeled with a probability $p_m$ of the measurement producing the wrong outcome. Gate noise is included with a two-qubit depolarizing channel with error probability $p_g$. In the simulations, we set $p_g=p_m$ and vary this noise probability parameter between 0.001 and 0.045. 

We have applied the heuristics in Section~\ref{sec:circuits} to find good circuits. We show our used circuits in Appendix \ref{sec:details_noisy_circuit_comparison}. In Fig.~\ref{fig:noisycircuit1}, we show how these circuits behave in the presence of operation noise versus circuits found with the genetic tools of~\cite{krastanov2019optimized}, for three different input fidelities of the initial Bell states $F_\textrm{in}$. Details about how the data is generated can be found in Appendix \ref{sec:details_noisy_circuit_comparison}. 

It is clear from Fig.~\ref{fig:noisycircuit1} that the genetic algorithm is more consistent in finding good protocols at $4\leq7$ than at $n=8$ and $n=9$. As explained in more detail in Appendix \ref{sec:details_noisy_circuit_comparison}, we used approximately 12 hours calculation time for each genetic algorithm data point. We expect that the $n=8$ and $n=9$ results become more consistent if one increases the calculation time. 

Furthermore, for each data point of the black-box method in Fig.~\ref{fig:noisycircuit1}, we plot a closed marker if the noiseless version of the circuit achieves the same distillation statistics as the protocol that achieves the highest fidelity in the case of no noise. Data points with an open marker have different distillation statistics without operation noise. From the results it becomes clear that, typically, at low $p_g=p_m$, the circuits found with~\cite{krastanov2019optimized} have the same distillation statistics as the best-performing noiseless circuits. At higher $p_g=p_m$, this is typically no longer the case: it is in this regime where the black-box method clearly outperforms the purely theoretical approach. This behaviour is not consistently present for $n=8$ and $n=9$: it might be that increasing the calculation time will show that protocols with the same distillation statistics as the optimal circuit with no operation noise will also work the best at low $p_g=p_m$ for $n=8$ and $n=9$. 

We now show the results for a 10 to 7 distillation protocol in Fig.~\ref{fig:noisycircuit4}. For the found 10 to 7 protocol we first found the $\left(n, k\right)$-graph that achieves the highest fidelity. Then, we applied random local complementations and edge flips to find an $\left(n,k\right)$-equivalent $\left(n, k\right)$-graph that would yield a low number of two-qubit gates and small number of keep-gates. We show our found representative and corresponding circuit in Figs.~\ref{fig:107graph} and~\ref{fig:107circuit}. As before, we find that for significant gate noise (i.e.~$p_g=p_m=0.05$) the black-box method achieves a higher fidelity. Furthermore, for $p_g=p_m=0.01$ both approaches perform comparable, with the heuristic optimization performing slightly better for lower input fidelities and worse for high input fidelities. We find in particular that the black-box algorithm cannot find the optimal protocol in the case of no noise.

\begin{figure}
\centerfloat
	\includegraphics[clip,  width=0.5\textwidth]{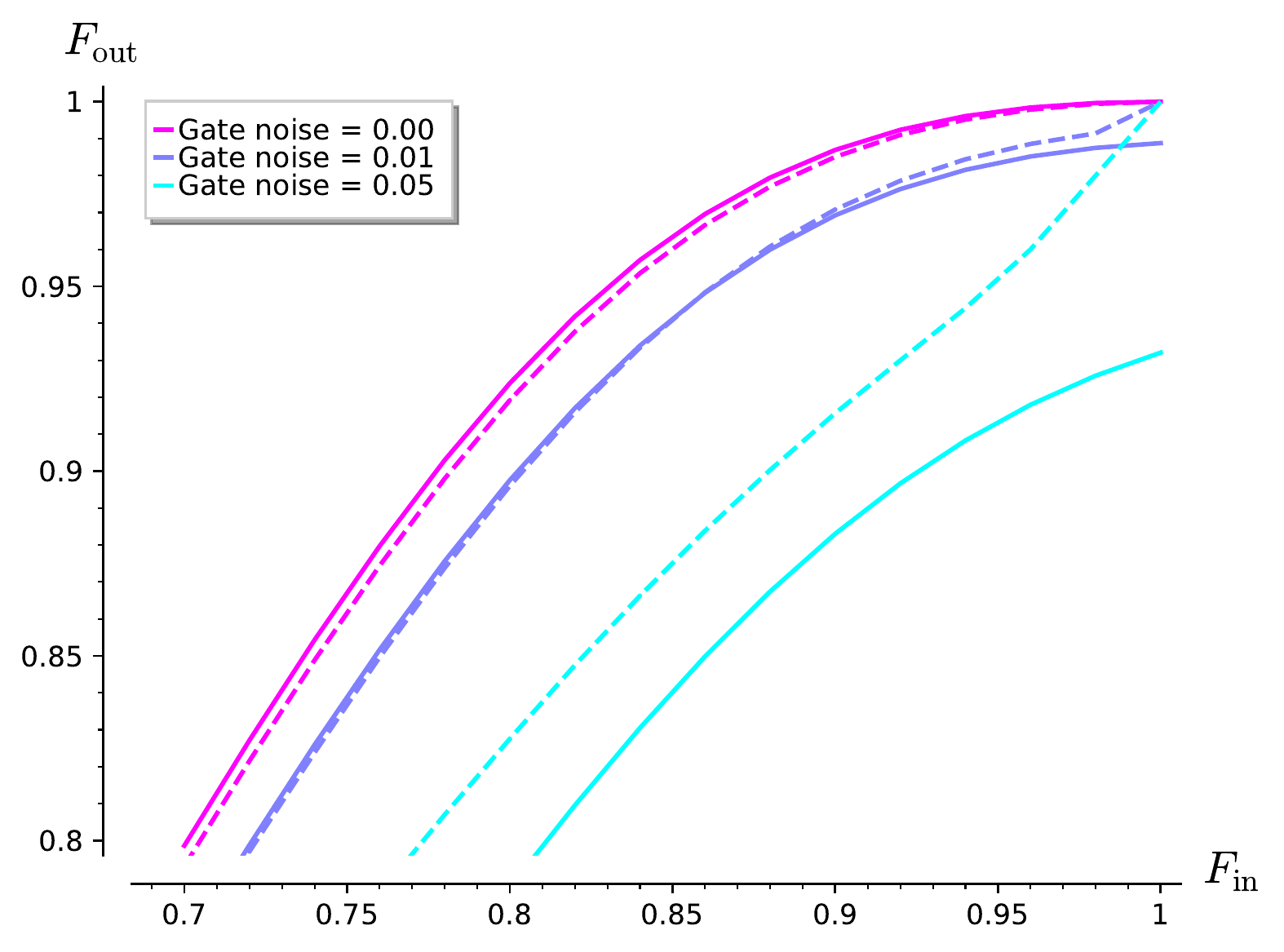}
	\vspace*{-3.2mm}
	\caption{Black-box optimization of circuits helps in the presence of local errors, but misleads if the local node hardware is perfect. This plots details the performance of circuits obtained from black-box optimization with genetic algorithms (dashed line) versus the performance of the best circuits obtained through graph enumeration (solid line). Output fidelity (vertical axis) is plotted against input fidelity (horizontal axis) at different local gate noise levels (color coded) for the best 10-to-7 purification circuit. The dashed lines do not correspond to a single circuit, rather at each parameter value we optimize a new circuit (e.g., at very high $F_\mathrm{in}$ the circuit is the trivial no-operation circuit, in order to avoid adding additional noise). For low gate noise parameters, the graph enumeration method discovers the best possible circuits and it outperforms the black-box method. On the other hand, the black-box method performs better in the presence of significant gate noise.}
	\label{fig:noisycircuit4}
\end{figure}

\subsection{Applying a $10$ to $7$ protocol to the teleportation of encoded states}
We now consider the teleportation of logical states between two users Alice and Bob. Teleportation ensures that the states are transmitted unconditionally, and the encoding increases the resilience against noise. As such, it can form a basis for quantum repeater schemes~\cite{munro2015inside}. We emphasize that, unlike the previous subsection, we consider here only the case of no noise on the gates in the circuits.

More concretely, Alice first creates a maximally entangled state, after which she encodes it into $2n'$ qubits using an error correction code $(n', 1, d')$ code. Then, she teleports one half of the state using $n'$ bipartite states shared with Bob. Finally, Bob decodes his share of the state. Here, $n$ to $k$ protocols with $k=n'>1$ could provide a potential benefit over the $k=1$ case, through reducing both the resultant infidelity and the number of initial states required. We use our tools for the case of the seven qubit Steane code~\cite{steane1996error} (i.e.~$n'=7$), for which we have found the $n=10$ to $k=7$ protocol with the highest fidelity, i.e.~the same one found in the previous subsection.

We compare this $10$ to $7$ protocol with two more standard approaches --- seven times the $2$ to $1$ DEJMPS protocol~\cite{Deutsch1996} and seven undistilled pairs. We compare the resultant (in)fidelities for several input fidelities in Fig.~\ref{fig:ntom}. We find that for input fidelities greater than $\approx 0.85$ the $10$ to $7$ protocol works best. Furthermore, taking into account the finite success probabilities of these protocols, we find that the $10$ to $7$ protocol requires less states on average than the seven times $2$ to $1$ protocol for input fidelities greater than $\approx0.95$, demonstrating the benefits of distillation protocols with $k>1$.

\begin{figure}
\centerfloat
	\includegraphics[clip,  width=0.48\textwidth, trim = 2.9mm 3.0mm 3.1mm 0mm]{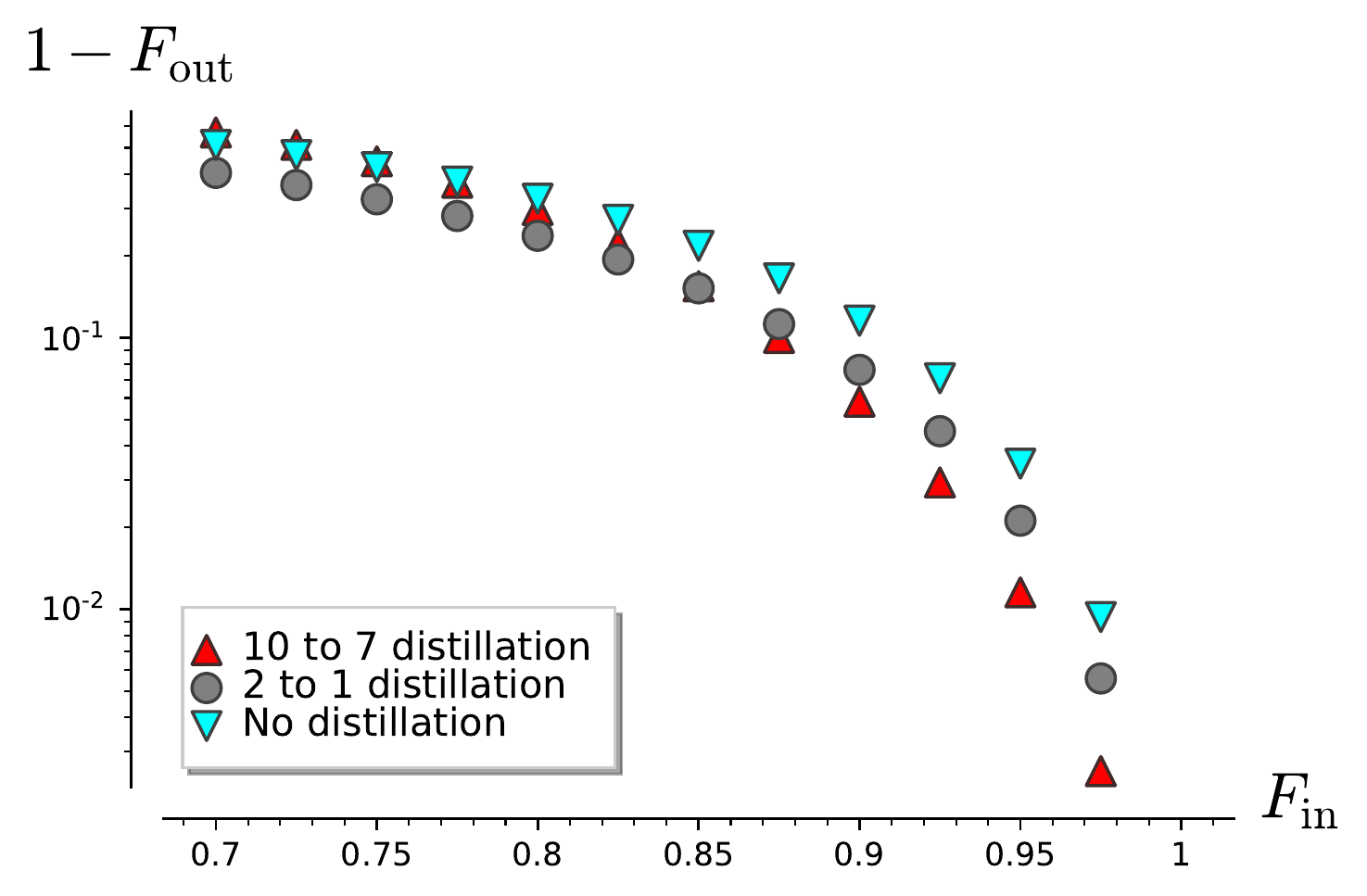}
	\vspace*{-3.2mm}
	\caption{Resultant infidelity after first teleporting one half of the logical maximally entangled state of the Steane code~\cite{steane1996error}, and then decoding the transmitted state. We find that teleportation using states from a $10$ to $7$ distillation protocol leads to an increase in fidelity for initial fidelities greater than $\approx 0.85$.}
	\label{fig:ntom}
\end{figure}

\section{Conclusions}\label{sec:conclusions}
In this work, we used a correspondence between stabilizer codes and bilocal Clifford protocols to reduce the search for distillation protocols to one over graphs. Furthermore, we found a way to map between such graphs and explicit circuits, allowing us to systematically construct distillation circuits with a small number of two-qubit gates and depth.

We have found that there is no distillation protocol (for fixed $n$ and $k$) that is optimal for a number of relevant quantities at the same time. That is, dependent on the quantity of interest and the input fidelity, different distillation protocols turned out to be optimal, highlighting the benefits of a full enumeration.
Moreover, we have shown that our results compare favorably with numerical optimization methods that explicitly take into account noise.

We have primarily focused here on the case of entanglement distillation. However, due to the correspondence between distillation and error correction, our enumeration can also be of interest to finding better quantum error correction protocols.

\section*{Acknowledgements}
We gratefully acknowledge support from the joint research programme “Modular quantum computers” by Fujitsu Limited and Delft University of Technology, co-funded by the Netherlands Enterprise Agency under project number PPS2007.

This work was supported by the Netherlands Organization for Scientific Research (NWO/OCW), as part of the Quantum Software Consortium program (Project No. 024.003.037/3368) and supported in part by the JST Moonshot R\&D program under Grant JPMJMS2061. The MIT Supercloud and Reed Fund provided valuable resources.
The authors thank Axel Dahlberg and Filip Rozp\k{e}dek for discussions and Tim Coopmans for feedback on the manuscript.

\newpage
\bibliography{references}

\begin{thebibliography}{10}

\bibitem{bennett1996purification}
C.~H. Bennett, G.~Brassard, S.~Popescu, B.~Schumacher, J.~A. Smolin, and W.~K.
  Wootters, ``Purification of noisy entanglement and faithful teleportation via
  noisy channels,'' {\em Physical review letters}, vol.~76, no.~5, p.~722,
  1996.

\bibitem{Bennett1996}
C.~H. Bennett, D.~P. Divincenzo, J.~A. Smolin, and W.~K. Wootters,
  ``Mixed-state entanglement and quantum error correction,'' {\em Physical
  Review A}, vol.~54, p.~3824–3851, Jan 1996.

\bibitem{Deutsch1996}
D.~Deutsch, A.~Ekert, R.~Jozsa, C.~Macchiavello, S.~Popescu, and A.~Sanpera,
  ``Quantum privacy amplification and the security of quantum cryptography over
  noisy channels,'' {\em Physical Review Letters}, vol.~77, pp.~2818--2821,
  Sept. 1996.

\bibitem{dur2007entanglement}
W.~D{\"u}r and H.~J. Briegel, ``Entanglement purification and quantum error
  correction,'' {\em Reports on Progress in Physics}, vol.~70, no.~8, p.~1381,
  2007.

\bibitem{krastanov2019optimized}
S.~Krastanov, V.~V. Albert, and L.~Jiang, ``Optimized entanglement
  purification,'' {\em Quantum}, vol.~3, p.~123, 2019.

\bibitem{Rozpdek2018}
F.~Rozp{\k{e}}dek, T.~Schiet, L.~P. Thinh, D.~Elkouss, A.~C. Doherty, and
  S.~Wehner, ``Optimizing practical entanglement distillation,'' {\em Physical
  Review A}, vol.~97, June 2018.

\bibitem{Bravyi2020}
S.~Bravyi and D.~Maslov, ``Hadamard-free circuits expose the structure of the
  {C}lifford group,'' {\em arXiv preprint arXiv:2003.09412v1 [quant-ph]}, Mar.
  2020.

\bibitem{fujii2009entanglement}
K.~Fujii and K.~Yamamoto, ``Entanglement purification with double selection,''
  {\em Physical Review A}, vol.~80, no.~4, p.~042308, 2009.

\bibitem{briegel1998quantum}
H.-J. Briegel, W.~D{\"u}r, J.~I. Cirac, and P.~Zoller, ``Quantum repeaters: the
  role of imperfect local operations in quantum communication,'' {\em Physical
  Review Letters}, vol.~81, no.~26, p.~5932, 1998.

\bibitem{dur2003entanglement}
W.~D{\"u}r and H.-J. Briegel, ``Entanglement purification for quantum
  computation,'' {\em Physical review letters}, vol.~90, no.~6, p.~067901,
  2003.

\bibitem{dur1999quantum}
W.~D{\"u}r, H.-J. Briegel, J.~I. Cirac, and P.~Zoller, ``Quantum repeaters
  based on entanglement purification,'' {\em Physical Review A}, vol.~59,
  no.~1, p.~169, 1999.

\bibitem{ruan2018adaptive}
L.~Ruan, W.~Dai, and M.~Z. Win, ``Adaptive recurrence quantum entanglement
  distillation for two-{K}raus-operator channels,'' {\em Physical Review A},
  vol.~97, p.~052332, May 2018.

\bibitem{vollbrecht2005interpolation}
K.~G.~H. Vollbrecht and F.~Verstraete, ``Interpolation of recurrence and
  hashing entanglement distillation protocols,'' {\em Physical Review A},
  vol.~71, no.~6, p.~062325, 2005.

\bibitem{jansen2020enum}
S.~Jansen, K.~Goodenough, S.~de~Bone, D.~Gijswijt, and D.~Elkouss,
  ``Enumerating all bilocal clifford distillation protocols through symmetry
  reduction,'' {\em Quantum}, vol.~6, p.~715, 2022.

\bibitem{schlingemann2001stabilizer}
D.~Schlingemann, ``Stabilizer codes can be realized as graph codes,'' {\em
  arXiv preprint quant-ph/0111080}, 2001.

\bibitem{gottesman1997stabilizer}
D.~Gottesman, {\em Stabilizer codes and quantum error correction}.
\newblock California Institute of Technology, 1997.

\bibitem{bouchet1988graphic}
A.~Bouchet, ``Graphic presentations of isotropic systems,'' {\em Journal of
  Combinatorial Theory, Series B}, vol.~45, no.~1, pp.~58--76, 1988.

\bibitem{wilde2011classical}
M.~M. Wilde, {\em Quantum information theory}.
\newblock Cambridge University Press, 2013.

\bibitem{bennett1996mixed}
C.~H. Bennett, D.~P. DiVincenzo, J.~A. Smolin, and W.~K. Wootters,
  ``Mixed-state entanglement and quantum error correction,'' {\em Physical
  Review A}, vol.~54, no.~5, p.~3824, 1996.

\bibitem{de2011symplectic}
M.~A. De~Gosson, {\em Symplectic methods in harmonic analysis and in
  mathematical physics}, vol.~7.
\newblock Springer Science \& Business Media, 2011.

\bibitem{hein2006entanglement}
M.~Hein, W.~D{\"u}r, J.~Eisert, R.~Raussendorf, M.~Nest, and H.-J. Briegel,
  ``Entanglement in graph states and its applications,'' {\em arXiv preprint
  quant-ph/0602096}, 2006.

\bibitem{hein2004multiparty}
M.~Hein, J.~Eisert, and H.~J. Briegel, ``Multiparty entanglement in graph
  states,'' {\em Physical Review A}, vol.~69, no.~6, p.~062311, 2004.

\bibitem{aschauer2005quantum}
H.~Aschauer, {\em Quantum communication in noisy environments}.
\newblock PhD thesis, lmu, 2005.

\bibitem{yu2007graphical}
S.~Yu, Q.~Chen, and C.~H. Oh, ``Graphical quantum error-correcting codes,''
  {\em arXiv preprint arXiv:0709.1780}, 2007.

\bibitem{cafaro2014scheme}
C.~Cafaro, D.~Markham, and P.~van Loock, ``Scheme for constructing graphs
  associated with stabilizer quantum codes,'' {\em arXiv preprint
  arXiv:1407.2777}, 2014.

\bibitem{hwang2015relation}
Y.~Hwang and J.~Heo, ``On the relation between a graph code and a graph
  state,'' {\em arXiv preprint arXiv:1511.05647}, 2015.

\bibitem{grassl2002graphs}
M.~Grassl, A.~Klappenecker, and M.~Rotteler, ``Graphs, quadratic forms, and
  quantum codes,'' in {\em Proceedings IEEE International Symposium on
  Information Theory,}, p.~45, IEEE, 2002.

\bibitem{englbrecht2022transformations}
M.~Englbrecht, T.~Kraft, and B.~Kraus, ``Transformations of stabilizer states
  in quantum networks,'' {\em arXiv preprint arXiv:2203.04202}, 2022.

\bibitem{cabello2011optimal}
A.~Cabello, L.~E. Danielsen, A.~J. L{\'o}pez-Tarrida, and J.~R. Portillo,
  ``Optimal preparation of graph states,'' {\em Physical Review A}, vol.~83,
  no.~4, p.~042314, 2011.

\bibitem{danielsen2006classification}
L.~E. Danielsen and M.~G. Parker, ``On the classification of all self-dual
  additive codes over gf (4) of length up to 12,'' {\em Journal of
  Combinatorial Theory, Series A}, vol.~113, no.~7, pp.~1351--1367, 2006.

\bibitem{glynn2004geometry}
D.~G. Glynn, T.~A. Gulliver, J.~G. Maks, and M.~K. Gupta, ``The geometry of
  additive quantum codes,'' {\em submitted to Springer-Verlag}, 2004.

\bibitem{van2005local}
M.~Van~den Nest, J.~Dehaene, and B.~De~Moor, ``Local unitary versus local
  clifford equivalence of stabilizer states,'' {\em Physical Review A},
  vol.~71, no.~6, p.~062323, 2005.

\bibitem{zeng2007local}
B.~Zeng, H.~Chung, A.~W. Cross, and I.~L. Chuang, ``Local unitary versus local
  clifford equivalence of stabilizer and graph states,'' {\em Physical Review
  A}, vol.~75, no.~3, p.~032325, 2007.

\bibitem{ji2007lu}
Z.~Ji, J.~Chen, Z.~Wei, and M.~Ying, ``The lu-lc conjecture is false,'' {\em
  arXiv preprint arXiv:0709.1266}, 2007.

\bibitem{dahlberg2020transform}
A.~Dahlberg, J.~Helsen, and S.~Wehner, ``How to transform graph states using
  single-qubit operations: computational complexity and algorithms,'' {\em
  Quantum Science and Technology}, vol.~5, no.~4, p.~045016, 2020.

\bibitem{munro2015inside}
W.~J. Munro, K.~Azuma, K.~Tamaki, and K.~Nemoto, ``Inside quantum repeaters,''
  {\em IEEE Journal of Selected Topics in Quantum Electronics}, vol.~21, no.~3,
  pp.~78--90, 2015.

\bibitem{shor1997quantum}
P.~Shor and R.~Laflamme, ``Quantum analog of the macwilliams identities for
  classical coding theory,'' {\em Physical review letters}, vol.~78, no.~8,
  p.~1600, 1997.

\bibitem{bennett2020quantum}
C.~H. Bennett and G.~Brassard, ``Quantum cryptography: Public key distribution
  and coin tossing,'' {\em Theoretical Computer Science}, vol.~560, pp.~7--11,
  2014.
\newblock Theoretical Aspects of Quantum Cryptography – celebrating 30 years
  of BB84.

\bibitem{jing2020quantum}
Y.~Jing, D.~Alsina, and M.~Razavi, ``Quantum key distribution over quantum
  repeaters with encoding: Using error detection as an effective postselection
  tool,'' {\em Physical Review Applied}, vol.~14, no.~6, p.~064037, 2020.

\bibitem{preskill2018quantum}
J.~Preskill, ``Quantum computing in the nisq era and beyond,'' {\em Quantum},
  vol.~2, p.~79, 2018.

\bibitem{steane1996error}
A.~M. Steane, ``Error correcting codes in quantum theory,'' {\em Physical
  Review Letters}, vol.~77, no.~5, p.~793, 1996.

\bibitem{addala2023inprep}
V.~Addala, S.~Ge, and S.~Krastanov {\em In preparation}, 2023.

\end{thebibliography}

\appendix

\section{$(n, k)$-graph and corresponding circuit for the found $10-7$ protocol}
We show in Fig.~\ref{fig:107graph} the $\left(n, k\right)$-graph found with our tools. First, we found an $\left(n, k\right)$-graph corresponding to a protocol that achieves the highest fidelity for $10$ to $7$ distillation. Then, we searched through the corresponding equivalence by applying random local complementations and edge flips to find an $\left(n, k\right)$-graph that would lead to the same fidelity, but a better circuit. The corresponding circuit found is shown in Fig.~\ref{fig:107circuit}. This circuit has $15$ two-qubit gates and depth $6$.

\begin{figure}[h!]
\centerfloat

\begin{tikzpicture}

\def\rb{4}
\def\ra{4.5}

\node[circle, fill=black, draw, scale=0.8] (1) at ({cos(0*360/17+90+7*360/17)*\rb}, {sin(0*360/17+90+7*360/17)*\rb}){};
\node[circle, fill=black, draw, scale=0.8] (2) at ({cos(1*360/17+90+7*360/17)*\rb}, {sin(1*360/17+90+7*360/17)*\rb}){};
\node[circle, fill=black, draw, scale=0.8] (3) at ({cos(2*360/17+90+7*360/17)*\rb}, {sin(2*360/17+90+7*360/17)*\rb}){};
\node[circle, fill=black, draw, scale=0.8] (4) at ({cos(3*360/17+90+7*360/17)*\rb}, {sin(3*360/17+90+7*360/17)*\rb}){};
\node[circle, fill=black, draw, scale=0.8] (5) at ({cos(4*360/17+90+7*360/17)*\rb}, {sin(4*360/17+90+7*360/17)*\rb}){};
\node[circle, fill=black, draw, scale=0.8] (6) at ({cos(5*360/17+90+7*360/17)*\rb}, {sin(5*360/17+90+7*360/17)*\rb}){};
\node[circle, fill=black, draw, scale=0.8] (7) at ({cos(6*360/17+90+7*360/17)*\rb}, {sin(6*360/17+90+7*360/17)*\rb}){};
\node[circle, fill=black, draw, scale=0.8] (8) at ({cos(7*360/17+90+7*360/17)*\rb}, {sin(7*360/17+90+7*360/17)*\rb}){};
\node[circle, fill=black, draw, scale=0.8] (9) at ({cos(8*360/17+90+7*360/17)*\rb}, {sin(8*360/17+90+7*360/17)*\rb}){};
\node[circle, fill=black, draw, scale=0.8] (10) at ({cos(9*360/17+90+7*360/17)*\rb}, {sin(9*360/17+90+7*360/17)*\rb}){};
\node[diamond, fill=black, draw, scale=0.8] (11) at ({cos(10*360/17+90+7*360/17)*\rb}, {sin(10*360/17+90+7*360/17)*\rb}){};
\node[diamond, fill=black, draw, scale=0.8] (12) at ({cos(11*360/17+90+7*360/17)*\rb}, {sin(11*360/17+90+7*360/17)*\rb}){};
\node[diamond, fill=black, draw, scale=0.8] (13) at ({cos(12*360/17+90+7*360/17)*\rb}, {sin(12*360/17+90+7*360/17)*\rb}){};
\node[diamond, fill=black, draw, scale=0.8] (14) at ({cos(13*360/17+90+7*360/17)*\rb}, {sin(13*360/17+90+7*360/17)*\rb}){};
\node[diamond, fill=black, draw, scale=0.8] (15) at ({cos(14*360/17+90+7*360/17)*\rb}, {sin(14*360/17+90+7*360/17)*\rb}){};
\node[diamond, fill=black, draw, scale=0.8] (16) at ({cos(15*360/17+90+7*360/17)*\rb}, {sin(15*360/17+90+7*360/17)*\rb}){};
\node[diamond, fill=black, draw, scale=0.8] (17) at ({cos(16*360/17+90+7*360/17)*\rb}, {sin(16*360/17+90+7*360/17)*\rb}){};

\node[scale=1] (10n) at ({cos(0*360/17+90+7*360/17)*\ra}, {sin(0*360/17+90+7*360/17)*\ra}){$1$};
\node[scale=1] (10n) at ({cos(1*360/17+90+7*360/17)*\ra}, {sin(1*360/17+90+7*360/17)*\ra}){$2$};
\node[scale=1] (9n) at ({cos(2*360/17+90+7*360/17)*\ra}, {sin(2*360/17+90+7*360/17)*\ra}){$3$};
\node[scale=1] (8n) at ({cos(3*360/17+90+7*360/17)*\ra}, {sin(3*360/17+90+7*360/17)*\ra}){$4$};
\node[scale=1] (7n) at ({cos(4*360/17+90+7*360/17)*\ra}, {sin(4*360/17+90+7*360/17)*\ra}){$5$};
\node[scale=1] (6n) at ({cos(5*360/17+90+7*360/17)*\ra}, {sin(5*360/17+90+7*360/17)*\ra}){$6$};
\node[scale=1] (5n) at ({cos(6*360/17+90+7*360/17)*\ra}, {sin(6*360/17+90+7*360/17)*\ra}){$7$};
\node[scale=1] (4n) at ({cos(7*360/17+90+7*360/17)*\ra}, {sin(7*360/17+90+7*360/17)*\ra}){$8$};
\node[scale=1] (3n) at ({cos(8*360/17+90+7*360/17)*\ra}, {sin(8*360/17+90+7*360/17)*\ra}){$9$};
\node[scale=1] (2n) at ({cos(9*360/17+90+7*360/17)*\ra}, {sin(9*360/17+90+7*360/17)*\ra}){$10$};
\node[scale=1] (1n) at ({cos(10*360/17+90+7*360/17)*\ra}, {sin(10*360/17+90+7*360/17)*\ra}){$1$};
\node[scale=1] (0n) at ({cos(11*360/17+90+7*360/17)*\ra}, {sin(11*360/17+90+7*360/17)*\ra}){$2$};
\node[scale=1] (-1n) at ({cos(12*360/17+90+7*360/17)*\ra}, {sin(12*360/17+90+7*360/17)*\ra}){$3$};
\node[scale=1] (-2n) at ({cos(13*360/17+90+7*360/17)*\ra}, {sin(13*360/17+90+7*360/17)*\ra}){$4$};
\node[scale=1] (-3n) at ({cos(14*360/17+90+7*360/17)*\ra}, {sin(14*360/17+90+7*360/17)*\ra}){$5$};
\node[scale=1] (-4n) at ({cos(15*360/17+90+7*360/17)*\ra}, {sin(15*360/17+90+7*360/17)*\ra}){$6$};
\node[scale=1] (-5n) at ({cos(16*360/17+90+7*360/17)*\ra}, {sin(16*360/17+90+7*360/17)*\ra}){$7$};



\def\rb{2}




\draw[line width = 0.3mm] (1) -- (10);
\draw[line width = 0.3mm] (1) -- (11);
\draw[line width = 0.3mm] (1) -- (12);
\draw[line width = 0.3mm] (1) -- (13);
\draw[line width = 0.3mm] (1) -- (14);
\draw[line width = 0.3mm] (1) -- (17);

\draw[line width = 0.3mm] (2) -- (6);
\draw[line width = 0.3mm] (2) -- (12);
\draw[line width = 0.3mm] (2) -- (13);
\draw[line width = 0.3mm] (2) -- (14);
\draw[line width = 0.3mm] (2) -- (17);
\draw[line width = 0.3mm] (2) -- (17);

\draw[line width = 0.3mm] (3) -- (6);
\draw[line width = 0.3mm] (3) -- (9);
\draw[line width = 0.3mm] (3) -- (13);
\draw[line width = 0.3mm] (3) -- (14);
\draw[line width = 0.3mm] (3) -- (15);

\draw[line width = 0.3mm] (4) -- (9);
\draw[line width = 0.3mm] (4) -- (12);
\draw[line width = 0.3mm] (4) -- (14);
\draw[line width = 0.3mm] (4) -- (15);
\draw[line width = 0.3mm] (4) -- (17);

\draw[line width = 0.3mm] (5) -- (7);
\draw[line width = 0.3mm] (5) -- (8);
\draw[line width = 0.3mm] (5) -- (11);
\draw[line width = 0.3mm] (5) -- (12);
\draw[line width = 0.3mm] (5) -- (15);
\draw[line width = 0.3mm] (5) -- (16);
\draw[line width = 0.3mm] (5) -- (17);

\draw[line width = 0.3mm] (6) -- (12);
\draw[line width = 0.3mm] (6) -- (16);

\draw[line width = 0.3mm] (7) -- (8);
\draw[line width = 0.3mm] (7) -- (12);
\draw[line width = 0.3mm] (7) -- (17);

\draw[line width = 0.3mm] (8) -- (14);
\draw[line width = 0.3mm] (8) -- (15);
\draw[line width = 0.3mm] (8) -- (16);
\draw[line width = 0.3mm] (8) -- (17);

\draw[line width = 0.3mm] (9) -- (15);
\draw[line width = 0.3mm] (9) -- (17);

\draw[line width = 0.3mm] (10) -- (12);
\draw[line width = 0.3mm] (10) -- (13);
\draw[line width = 0.3mm] (10) -- (14);
\draw[line width = 0.3mm] (10) -- (15);
\draw[line width = 0.3mm] (10) -- (17);


\end{tikzpicture}
\caption{The found $\left(n, k\right)$-graph that yields a distillation protocol with the highest fidelity for $10$ to $7$ distillation. The numbering is added to correspond to the lines in Fig.~\ref{fig:107circuit}}
\label{fig:107graph}
\end{figure}
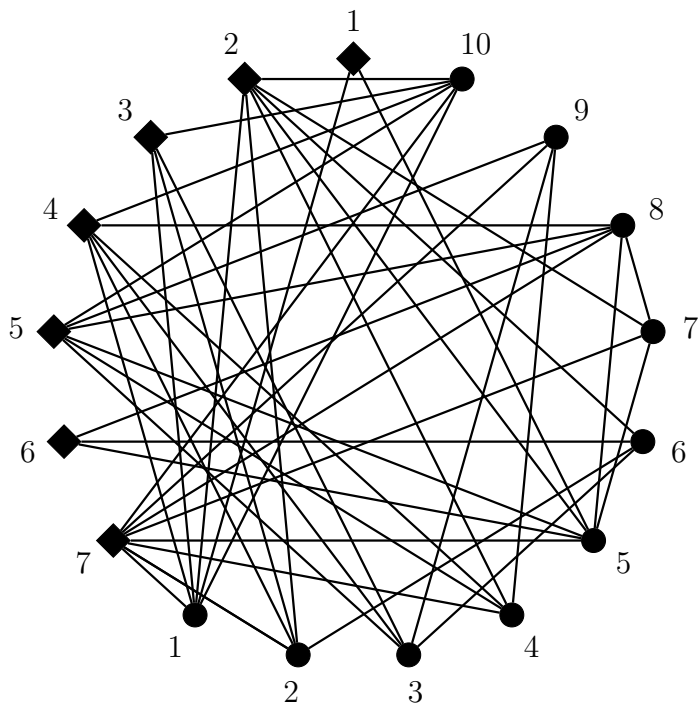

\begin{figure}[h!]
\centerfloat
\input{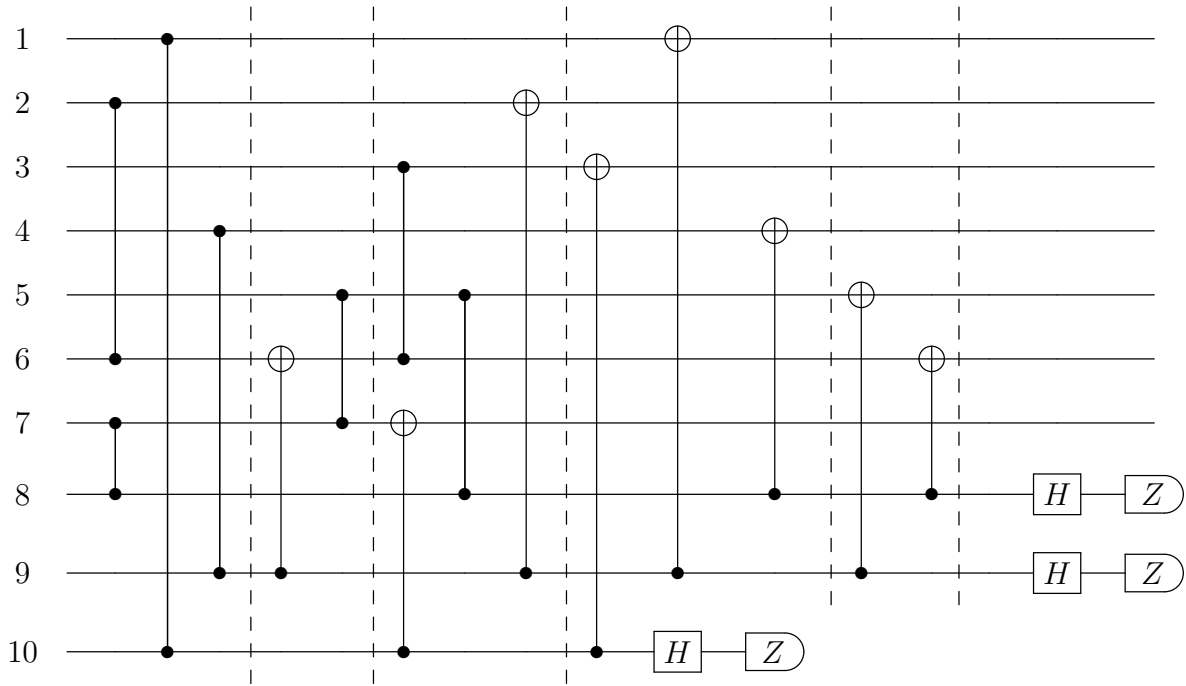}
\caption{Corresponding circuit constructed from the $\left(n, k\right)$-graph in Fig.~\ref{fig:107graph}. This circuit has depth $6$ (where each time-step is demarcated by the dashed lines) and $15$ two-qubit gates.}
\label{fig:107circuit}
\end{figure}

\section{Finding transversals using extensions}\label{sec:extensions}
Here we detail an approach --- similar to work from \cite{danielsen2006classification, glynn2004geometry} --- on how to more efficiently find transversals under the local equivalence relation on $\left(n, k\right)$-graphs. We do so by using \emph{extensions}.
First, let $G$ be any $\left(n, k\right)$-graph. An \emph{output extension} of an $\left(n, k\right)$-graph $G$ is any of the possible $\left(n+1, k\right)$ graphs obtained by adjoining an isolated output vertex to $G$, and adding at least one of the possible $n+k$ edges from the isolated vertex to any of the other $n+k$ vertices. An \emph{input extension} of an $\left(n, k\right)$-graph $G$ is any of the possible $\left(n, k+1\right)$ graphs obtained by adjoining an isolated input vertex to $G$, and adding at least one of the possible $n$ edges from the isolated vertex to any of the $n$ output vertices.

\begin{lemma}
Let $\mathbb{L}_{n}^{k}$ be an arbitrary transversal of connected graphs under the local equivalence relation on $\left(n, k\right)$-graphs. The set of size $\left(2^{n+k}-1\right)\left|\mathbb{L}_{n}^{k}\right|$ obtained by performing an output extension on every graph in $\mathbb{L}_{n}^{k}$ contains a transversal of graphs under the local equivalence relation on $\left(n+1, k\right)$-graphs. Furthermore, the set of size  $\left(2^{n}-1\right)\left|\mathbb{L}_{n}^{k}\right|$ obtained from performing an input extension on every graph in $\mathbb{L}_{n}^{k}$ contains a transversal of graphs under the local equivalence relation on $\left(n, k+1\right)$-graphs.
\end{lemma}
\begin{proof}
The proof follows the same logic as that in \cite{danielsen2006classification, glynn2004geometry}. First, let $\mathbb{L}_{n+1}^{k}$ be an arbitrary transversal of the local equivalence relation on $\left(n+1, k\right)$-graphs, and choose an arbitrary $\left(n+1, k\right)$-graph $G$. From the $n+1+k$ vertices of $G$, choose an arbitrary subset $V'$ that excludes exactly one of the output vertices. Since the induced subgraph $G[V']$ is an $\left(n, k\right)$-graph, it is possible to perform local complementations on the vertices in $V'$, together with edge flips on the $k$ input vertices such that $G[V']$ is equivalent up to an $\left(n, k\right)$-permutation to some representative $G' \in \mathbb{L}_{n}^{k}$. But then $G$ is equivalent up to an $\left(n,k\right)$-permutation to an extension of $G'$.
A similar argument holds for input extensions, but now an arbitrary $\left(n,k+1\right)$-graph $G$ is chosen and $V'$ is a subset that excludes one input qubit. An input extension of the induced subgraph $G[V']$ is then equivalent up to $\left(n, k\right)$-permutations and edge flips to $G$.
\end{proof}

\begin{figure}[h!]
\centerfloat
\begin{subfigure}{0.15\textwidth}

\begin{tikzpicture}

\node[circle, fill=black, draw, scale=0.6] (a) at (-1,+1){};
\node[circle, fill=black, draw, scale=0.6] (b) at (+1, +1){};
\node[circle, fill=black, draw, scale=0.6] (c) at (+1, -1){};
\node[circle, fill=black, draw, scale=0.6] (d) at (-1, -1){};

\node[diamond, fill=black, draw, scale=0.6] (d1) at (0.6, 2.5){};
\node[diamond, fill=black, draw, scale=0.6] (d2) at (-0.6, 2.5){};


\draw[line width = 0.3mm] (1,1) -- (1,-1) -- (-1,-1) -- (-1,+1);

\draw[line width = 0.3mm] (-0.6, 2.5) -- (-1,1);
\draw[line width = 0.3mm] (-0.6, 2.5) -- (1,-1);
\draw[line width = 0.3mm] (-0.6, 2.5) -- (1,1);

\draw[line width = 0.3mm] (0.6, 2.5) -- (-1,-1);
\draw[line width = 0.3mm] (0.6, 2.5) -- (1,-1);

\draw[line width = 0.3mm] (0.6, 2.5) -- (-0.6,2.5);

\end{tikzpicture}
\end{subfigure}
\begin{subfigure}{0.15\textwidth}

\begin{tikzpicture}

\node[circle, fill=black, draw, scale=0.6] (a) at (-1,+1){};
\node[circle, fill=black, draw, scale=0.6] (b) at (+1, +1){};
\node[circle, fill=black, draw, scale=0.6] (c) at (+1, -1){};
\node[circle, fill=black, draw, scale=0.6] (d) at (-1, -1){};

\node[circle, fill=black, draw, scale=0.6] (d) at (0, 0){};

\draw[dashed, line width = 0.3mm] (0,0) -- (1,-1);
\draw[dashed, line width = 0.3mm] (0,0) -- (1,+1);
\draw[dashed, line width = 0.3mm] (0,0) -- (-1,-1);
\draw[dashed, line width = 0.3mm] (0,0) -- (-1,+1);
\draw[dashed, line width = 0.3mm] (0,0) -- (0.6, 2.5);
\draw[dashed, line width = 0.3mm] (0,0) -- (-0.6, 2.5);

\node[diamond, fill=black, draw, scale=0.6] (d1) at (0.6, 2.5){};
\node[diamond, fill=black, draw, scale=0.6] (d2) at (-0.6, 2.5){};


\draw[line width = 0.3mm] (1,1) -- (1,-1) -- (-1,-1) -- (-1,+1);

\draw[line width = 0.3mm] (-0.6, 2.5) -- (-1,1);
\draw[line width = 0.3mm] (-0.6, 2.5) -- (1,-1);
\draw[line width = 0.3mm] (-0.6, 2.5) -- (1,1);

\draw[line width = 0.3mm] (0.6, 2.5) -- (-1,-1);
\draw[line width = 0.3mm] (0.6, 2.5) -- (1,-1);

\draw[line width = 0.3mm] (0.6, 2.5) -- (-0.6,2.5);

\end{tikzpicture}
\end{subfigure}
\begin{subfigure}{0.15\textwidth}

\begin{tikzpicture}

\node[circle, fill=black, draw, scale=0.6] (a) at (-1,+1){};
\node[circle, fill=black, draw, scale=0.6] (b) at (+1, +1){};
\node[circle, fill=black, draw, scale=0.6] (c) at (+1, -1){};
\node[circle, fill=black, draw, scale=0.6] (d) at (-1, -1){};

\node[diamond, fill=black, draw, scale=0.6] (d) at (0, 0){};

\draw[dashed, line width = 0.3mm] (0,0) -- (1,-1);
\draw[dashed, line width = 0.3mm] (0,0) -- (1,+1);
\draw[dashed, line width = 0.3mm] (0,0) -- (-1,-1);
\draw[dashed, line width = 0.3mm] (0,0) -- (-1,+1);

\node[diamond, fill=black, draw, scale=0.6] (d1) at (0.6, 2.5){};
\node[diamond, fill=black, draw, scale=0.6] (d2) at (-0.6, 2.5){};


\draw[line width = 0.3mm] (1,1) -- (1,-1) -- (-1,-1) -- (-1,+1);

\draw[line width = 0.3mm] (-0.6, 2.5) -- (-1,1);
\draw[line width = 0.3mm] (-0.6, 2.5) -- (1,-1);
\draw[line width = 0.3mm] (-0.6, 2.5) -- (1,1);

\draw[line width = 0.3mm] (0.6, 2.5) -- (-1,-1);
\draw[line width = 0.3mm] (0.6, 2.5) -- (1,-1);

\draw[line width = 0.3mm] (0.6, 2.5) -- (-0.6,2.5);

\end{tikzpicture}
\end{subfigure}
\caption{Left) An $\left(n, k\right) = \left(4, 2\right)$ graph on $n+k = 4+2$ vertices. Middle) an extension of the first graph. Right) An input extension of the first graph. Possible edges are indicated by dashed lines. Input vertices are indicated by diamond nodes.}
\label{fig:extensions}
\end{figure}
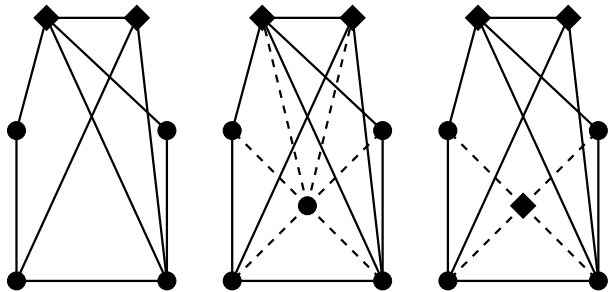

\section{Machine learning approach}

The main body of this work deals with first-principles, analytical, efficient enumeration of good purification protocols. However, this approach does not automatically provide the best circuit implementing a given protocol, neither does it consider the detrimental effects of imperfect local gates. We used alternative tools in order to study how effective our approach is when considering the aforementioned additional constraints. Namely, we employed a known black box optimizer for the generation of good noisy purification circuits~\cite{krastanov2019optimized}, albeit without optimality guarantees. This black box optimizer consists of two parts: a noisy entanglement simulator and a genetic optimization algorithm.

The simulator works by restricting the representation of the Bell pairs to only states that can be expressed as density matrices diagonal in the Bell basis.
Gates in the purification protocols are simply permutations of the Bell basis and measurements are simply deletion of half of the basis states, thus providing for very efficient simulation (faster than Clifford circuit simulation). Our particular simulator is exponentially costly in the number of Bell pairs due to purely classical reasons: we track all possible correlations between Bell-diagonal states. However if that becomes a practical problem, a standard classical Monte Carlo approach would be enough to speed up the simulation at a fairly modest cost to the precision of the simulation results (as we do in a yet to be published related work~\cite{addala2023inprep}).

The genetic algorithm employed for the simulation is fairly conventional: we represent circuits as a sequence of gates. That sequence forms the "genome" of the circuit. Each circuit is an "individual" in a large "population" of circuits. At each iteration of the optimization algorithm we generate "offspring" circuits by randomly mixing up the genome of "parent" circuits. At each iteration we also generate "mutant" circuits by randomly perturbing existing circuits. Random perturbation can be anything from swapping the order of a pair of gates, to changing the parameters of a gate (e.g. a CNOT becomes a CPHASE). This new "generation" of circuits is evaluated and the worst performers are culled. The procedure is repeated until we converge on good circuits, which usually takes a hundred generations and less than an hour on commodity hardware for registers of width under 8 qubits.

The only gates permitted in the genome are gates that map "good" Bell pairs to the same Bell pair, but permute the other possible basis states arbitrarily.

\section{Details noisy circuit comparison}\label{sec:details_noisy_circuit_comparison}

In Sec. \ref{sec:noisy_circuit_comparison} and Fig. \ref{fig:noisycircuit1} of the main text, we compare protocols found with our heuristic method to protocols found with the genetic tools of~\cite{krastanov2019optimized} in situations with gate and measurement noise. Here, we will provide details on how the data of Fig. \ref{fig:noisycircuit1} is generated.

Because we wanted to compare our results to the circuits generated with~\cite{krastanov2019optimized} for specific Bell state numbers $n$, we had to slightly adjust the code of~\cite{krastanov2019optimized}. In creating the new generation of circuits, we introduced a check that made sure if the number of `raw' (\textit{i.e.}, input) Bell pairs used for the specific individual circuit did not exceed $n$. This adjustment is very similar to the already existing check in the code that made sure the number of total operations does not exceed a preset number.

To generate the results, we set the number of register qubits of the circuits to $n$. Strictly speaking, one could also generate circuits for a certain number of input Bell states $n$ with a smaller register, as the circuits re-use measured-out qubits. However, to make sure we would not exclude distillation circuits, we decided to use the maximum register size. For each of the initial individuals of the population, we selected $n+2$ random operations. During evolution, we let the number of gates and measurements grow or shrink without restrictions. We made use of a population size of 300 circuits. When creating children, we used 20 random pairs of this population, and generated 100 children for each pair. During mutation, per individual of the population, we generated 2 mutants for each of the 4 different mutant types included in the code.

We let the software generate a maximum of 100 generations, but also, for each data point of Fig. \ref{fig:noisycircuit1}, cut-off the creation of new generations after 12 hours. If all of the 100 generations were generated before the 12 hour mark, or if the population converged with a smaller number of generations before the 12 hour mark, we started a new iteration of the software with a new random starting population and a different seed. At the end, we selected the best result from all iterations.

\section{Selection of found circuits}
We present here some of the circuits found with our optimization. For each $n$, we selected the circuits based on the output fidelity of the final state at input state fidelity $F_\textrm{in}=0.9$ and operation noise $p_g=p_m=0.03$. 

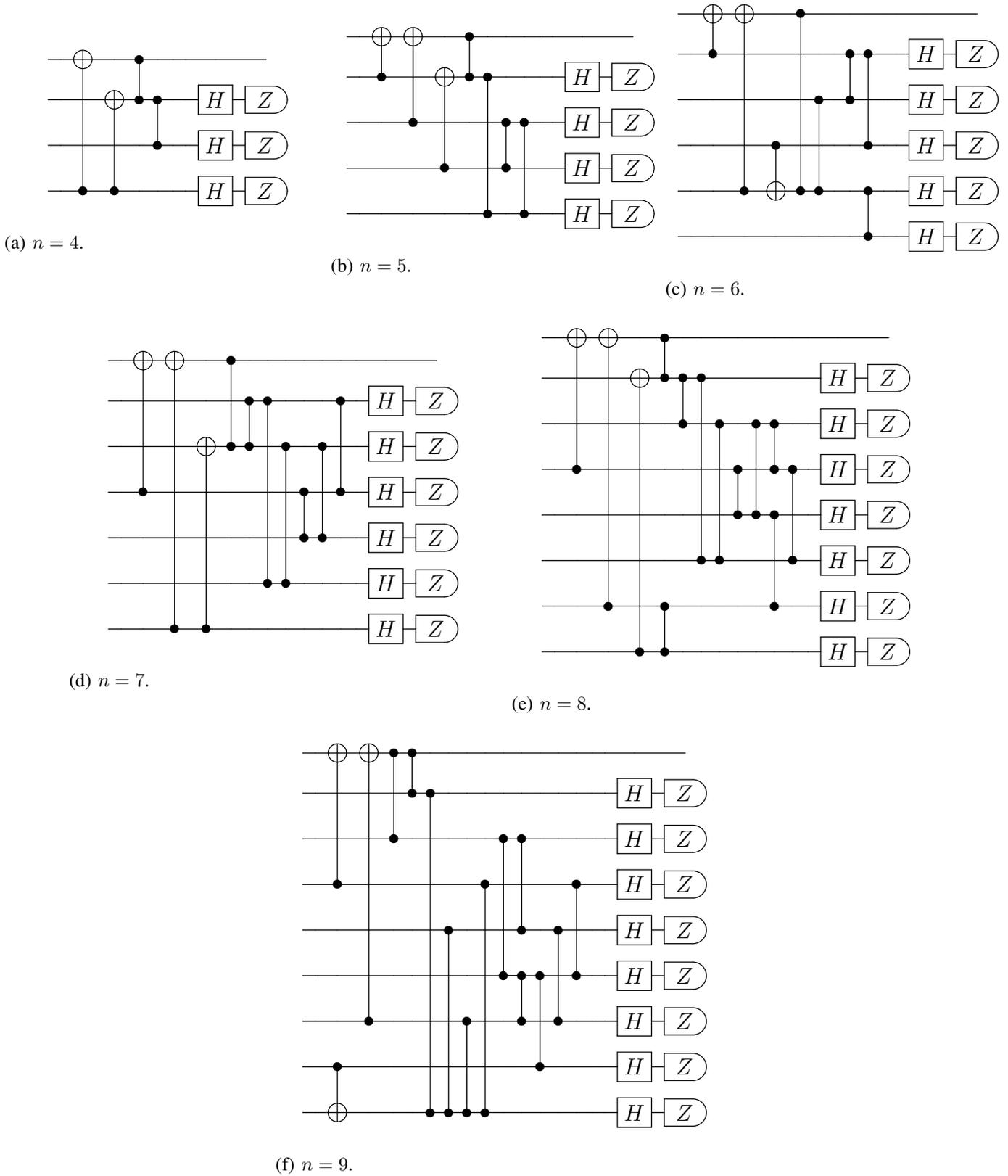
\begin{figure}
\centerfloat
\begin{subfigure}{.33\textwidth}
\centerfloat
\[
  \Qcircuit @C=.55em @R=.7em {
         	&\qw &\targ 	&\qw 	        &\ctrl{1}  		&\qw    	    &\qw  	&\qw    &\qw   	    &\qw			\\ 
        	&\qw &\qw       &\targ       &\control\qw    &\control\qw            &\qw  	&\qw	& \gate{H} & \measureD{Z}  	\\ 
        	&\qw &\qw  	    &\qw   &\qw          	&\ctrl{-1}	    &\qw  	&\qw 	& \gate{H} & \measureD{Z}	\\ 
        	&\qw &\ctrl{-3} &\ctrl{-2}    	    &\qw		    &\qw    &\qw  	&\qw  	& \gate{H} & \measureD{Z}		
} 
\]
\caption{$n=4$.}
\label{fig:opt_prot_n=4}
\end{subfigure}%
\begin{subfigure}{.33\textwidth}
\centerfloat
\[
  \Qcircuit @C=.55em @R=.7em {
        & 	&\qw 	&\targ	 	&\targ		&\qw		&\ctrl{1}	    &\qw		    &\qw		    &\qw	        &\qw	&\qw	&\qw        &\qw			\\ 
        &	&\qw  	&\ctrl{-1}		&\qw		&\targ		&\control\qw		    &\ctrl{3}	    &\qw	    &\qw	        &\qw	&\qw	& \gate{H}  &\measureD{Z}  	\\ 
        &	&\qw  	&\qw  &\ctrl{-2}  	    &\qw		&\qw	&\qw	        &\control\qw	&\ctrl{2}	        &\qw	&\qw 	& \gate{H}  &\measureD{Z}	\\ 
        &	&\qw  	&\qw	    &\qw  	    &\ctrl{-2}	&\qw		    &\qw		    &\ctrl{-1}\qw    		&\qw       &\qw	&\qw	& \gate{H}  &\measureD{Z}	\\ 
        &	&\qw  	&\qw 		&\qw 	&\qw	    &\qw		    &\control\qw	&\qw        	&\control\qw    &\qw	&\qw	& \gate{H}  &\measureD{Z}		
} 
\]
\caption{$n=5$.}
\label{fig:opt_prot_n=5}
\end{subfigure}
\begin{subfigure}{0.33\textwidth}
\centerfloat
\[
  \Qcircuit @C=.55em @R=.7em {
        & 	&\qw 	&\targ	 	&\targ		&\qw		&\ctrl{4}		&\qw		    &\qw		    &\qw	    &\qw		    &\qw &\qw   &\qw	    &\qw			\\ 
        &	&\qw  	&\ctrl{-1}		&\qw		&\qw		&\qw		&\qw		    &\qw	    &\ctrl{1}		    &\ctrl{2}	    &\qw &\qw	& \gate{H}  &\measureD{Z}  	\\ 
        &	&\qw  	&\qw	&\qw  &\qw		&\qw		&\control\qw		    &\qw	&\control\qw	&\qw		    &\qw &\qw   & \gate{H}  &\measureD{Z}	\\ 
        &	&\qw  	&\qw		&\qw	  	&\ctrl{1}		&\qw	&\qw		    &\qw		    &\qw		    &\control\qw		    &\qw &\qw   & \gate{H}  &\measureD{Z}	\\ 
        &	&\qw  	&\qw 		&\ctrl{-4} 		&\targ	&\control\qw		&\ctrl{-2}	    &\qw		    &\qw		    &\control\qw	&\qw &\qw   & \gate{H}  &\measureD{Z}	\\	
        &	&\qw  	&\qw		&\qw	  	&\qw		&\qw		&\qw	&\qw		    &\qw		    &\ctrl{-1}		    &\qw &\qw   & \gate{H}  &\measureD{Z}	 
} 
\]
\caption{$n=6$.}
\label{fig:opt_prot_n=6}
\end{subfigure}

\begin{subfigure}{0.42\textwidth}
\centerfloat
\[
  \Qcircuit @C=.55em @R=.7em {
        & 	&\qw 	&\targ	 	&\targ		&\qw		&\ctrl{2}		&\qw		    &\qw		    &\qw	    &\qw		    &\qw		    &\qw		    &\qw		    &\qw   &\qw	 			\\ 
        &	&\qw  	&\qw		&\qw		&\qw		&\qw	&\ctrl{1}		    &\ctrl{4}	    &\qw		    &\qw	    &\qw		    &\ctrl{2}		    &\qw			&\gate{H}  &\measureD{Z}  	\\ 
        &	&\qw  	&\qw  		&\qw  &\targ		&\control\qw		&\control\qw		    &\qw		    &\ctrl{3}	&\qw		    &\ctrl{2}	    &\qw    &\qw		   &\gate{H}  &\measureD{Z}	\\ 
        &	&\qw  	&\ctrl{-3}		&\qw	  	&\qw		&\qw		&\qw		    &\qw		    &\qw		    &\ctrl{1}	&\qw		    &\control\qw	&\qw		  &\gate{H}  &\measureD{Z}	\\ 
        &	&\qw  	&\qw	&\qw	  	&\qw	&\qw		&\qw		    &\qw		    &\qw		    &\control\qw		    &\control\qw		    &\qw		     &\qw   &\gate{H}  &\measureD{Z}	\\ 
        &	&\qw  	&\qw		&\qw	  	&\qw		&\qw		&\qw	    &\control\qw	&\control\qw		    &\qw		    &\qw		    &\qw		  	&\qw  &\gate{H}  &\measureD{Z}	\\ 
        &	&\qw  	&\qw 		&\ctrl{-6} 		&\ctrl{-4}		&\qw		&\qw	&\qw		    &\qw		    &\qw		    &\qw 	    &\qw &\qw   &\gate{H}  &\measureD{Z}		
} 
\]
\caption{$n=7$.}
\label{fig:opt_prot_n=7}
\end{subfigure}\hspace{5mm}%
\begin{subfigure}{0.42\textwidth}
\centerfloat
\[
  \Qcircuit @C=.55em @R=.7em {
        & 	&\qw 	&\targ	 	&\targ	  	&\qw		&\ctrl{1}		&\qw		&\qw		&\qw		&\qw		    &\qw	    &\qw		    &\qw		    &\qw		    &\qw		    &\qw 			\\ 
        &	&\qw  	&\qw		&\qw	  	&\targ  &\control\qw		&\ctrl{1}		&\ctrl{4}		&\qw		&\qw	    &\qw		    &\qw    	    &\qw	        &\qw  &\gate{H}   &\measureD{Z}	\\ 
        &	&\qw  	&\qw		&\qw	  	&\qw  &\qw		&\control\qw		&\qw		&\ctrl{3}		&\qw	    &\control\qw		    &\ctrl{1}    	    &\qw	        &\qw	  &\gate{H}   &\measureD{Z}	\\ 
              &	&\qw  	&\ctrl{-3}		&\qw	  	&\qw  &\qw		&\qw		&\qw		&\qw		&\ctrl{1}	    &\qw		    &\control\qw    	    &\control\qw	        &\qw&\gate{H}   &\measureD{Z}	\\ 
               &	&\qw  	&\qw		&\qw	  	&\qw  &\qw		&\qw		&\qw		&\qw		&\control\qw	    &\ctrl{-2}		    &\control\qw   	    &\qw &\qw   &\gate{H}   &\measureD{Z}	\\ 
               &	&\qw  	&\qw		&\qw	  	&\qw  &\qw		&\qw		&\control\qw		&\control\qw		&\qw	    &\qw		    &\qw    	    &\ctrl{-2}	   &\qw   &\gate{H}   &\measureD{Z}	\\ 
                &	&\qw  	&\qw		&\ctrl{-6}	  	&\qw  &\control\qw		&\qw		&\qw		&\qw		&\qw	    &\qw		    &\ctrl{-2}   	    &\qw &\qw   &\gate{H}   &\measureD{Z}	\\ 
               &	&\qw  	&\qw		&\qw	  	&\ctrl{-6}	  &\ctrl{-1}		&\qw		&\qw		&\qw		&\qw	    &\qw		    &\qw  		    &\qw &\qw   &\gate{H}   &\measureD{Z}	\\ 		
} 
\]
\caption{$n=8$.}
\label{fig:opt_prot_n=8}
\end{subfigure}

\begin{subfigure}{0.45\textwidth}
\centerfloat
\[
  \Qcircuit @C=.55em @R=.7em {
        & 	&\qw 	&\targ	 	&\targ	  	&\ctrl{2}		&\ctrl{1}		&\qw		&\qw		&\qw		&\qw		    &\qw	    &\qw		    &\qw		    &\qw		    &\qw		    &\qw &\qw   &\qw	    &\qw			\\ 
        &	&\qw  	&\qw		&\qw	  	&\qw  &\control\qw		&\ctrl{7}		&\qw		&\qw		&\qw	    &\qw		    &\qw    	    &\qw	        &\qw		&\qw		    &\qw &\qw   &\gate{H}   &\measureD{Z}	\\ 
        &	&\qw  	&\qw		&\qw	  	&\control\qw  &\qw		&\qw		&\qw		&\qw		&\qw	    &\ctrl{3}		    &\ctrl{2}    	    &\qw	        &\qw		&\qw		    &\qw &\qw   &\gate{H}   &\measureD{Z}	\\ 
              &	&\qw  	&\ctrl{-3}		&\qw	  	&\qw  &\qw		&\qw		&\qw		&\qw		&\ctrl{5}	    &\qw		    &\qw    	    &\qw	        &\qw		&\ctrl{2}		    &\qw &\qw   &\gate{H}   &\measureD{Z}	\\ 
               &	&\qw  	&\qw		&\qw	  	&\qw  &\qw		&\qw		&\ctrl{4}		&\qw		&\qw	    &\qw		    &\control\qw    	    &\qw	        &\control\qw		&\qw		    &\qw &\qw   &\gate{H}   &\measureD{Z}	\\ 
               &	&\qw  	&\qw		&\qw	  	&\qw  &\qw		&\qw		&\qw		&\qw		&\qw	    &\control\qw		    &\ctrl{1}\qw    	    &\control\qw	        &\qw		&\control\qw		    &\qw &\qw   &\gate{H}   &\measureD{Z}	\\ 
                &	&\qw  	&\qw		&\ctrl{-6}	  	&\qw  &\qw		&\qw		&\qw		&\ctrl{2}  		&\qw	    &\qw		    &\control\qw    	    &\qw       &\ctrl{-2}		&\qw		    &\qw &\qw   &\gate{H}   &\measureD{Z}	\\
               &	&\qw  	&\ctrl{1}		&\qw	  	&\qw  &\qw		&\qw		&\qw		&\qw		&\qw	    &\qw		    &\qw    	    &\ctrl{-2}	        &\qw		&\qw		    &\qw &\qw   &\gate{H}   &\measureD{Z}	\\
                              &	&\qw  	&\targ		&\qw	  	&\qw  &\qw		&\control\qw		&\control\qw		&\control\qw		&\control\qw	    &\qw		    &\qw    	    &\qw	        &\qw		&\qw		    &\qw &\qw   &\gate{H}   &\measureD{Z}	\\ 		
} 
\]
\caption{$n=9$.}
\label{fig:opt_prot_n=9}
\end{subfigure}
\caption{The circuits found through the heuristic optimization.}
\label{fig:opt_prots}
\end{figure}

\end{document}